\documentclass[letter,11pt]{article}

\usepackage[colorlinks]{hyperref}
  \hypersetup{linkcolor=blue,filecolor=blue,citecolor=blue,urlcolor=blue}

\usepackage{amsmath,amsthm}
\usepackage[T1]{fontenc}

  \newcommand{\fullversion}[1]{#1}
\newcommand{\submversion}[1]{}
  
\usepackage{xspace}
\usepackage{fullpage}

\usepackage{verbatim}

\newcommand{\bra}[1]{\langle #1|}
\newcommand{\ket}[1]{|#1\rangle}
\newcommand{\braket}[2]{\langle #1|#2\rangle}
\newcommand{\ketbra}[2]{|#1\rangle\langle #2|}

\usepackage{amsthm}
\newtheorem{theorem}{Theorem}

\newtheorem{definition}{Definition}
\newtheorem{lemma}{Lemma}
\newtheorem{claim}{Claim}
\newtheorem{remark}{Remark}
\newtheorem{corollary}{Corollary}

\usepackage{amsmath,amsfonts,amssymb}

\usepackage[utf8]{inputenc}
\usepackage{tabularx}
\usepackage[normalem]{ulem}
\usepackage{setspace}
\usepackage{fancyhdr}
\usepackage{lastpage}
\usepackage{extramarks}
\usepackage{chngpage}
\usepackage{soul}
\usepackage{xspace}
\usepackage{bm}
\usepackage{nicefrac}
\usepackage{paralist}
\usepackage{enumitem}
\usepackage{tikz}

\newcommand{\secparam}{\lambda}
\newcommand{\negl}{\mathsf{negl}}

\newcommand{\ignore}[1]{}

\submversion{

}
\fullversion{

}

\newcommand{\cktclass}{\mathcal{C}}

\newcommand{\poly}{\mathrm{poly}}

\newcommand{\distr}{\mathcal{D}}

\newcommand{\gen}{\mathsf{Gen}}

\newcommand{\sign}{\mathsf{Sign}}
\newcommand{\ver}{\mathsf{Ver}}

\newcommand{\ch}{\mathsf{Ch}}
\newcommand{\expt}{\mathsf{Expt}}
\newcommand{\adversary}{\mathcal{A}}  

\newcommand{\hybrid}{\mathsf{Hyb}} 
\newcommand{\prob}{\mathsf{Pr}} 

\newcommand{\setup}{\mathsf{Setup}}
\newcommand{\keygen}{\mathsf{KeyGen}}
\newcommand{\enc}{\mathsf{Enc}}
\newcommand{\dec}{\mathsf{Dec}}

\newcommand{\msk}{\mathsf{MSK}}
\newcommand{\pk}{\mathsf{PK}}
\newcommand{\ct}{\mathsf{CT}}

\newcommand{\pke}{\mathsf{PKE}}

\newcommand{\fe}{\mathsf{FE}}

\newcommand{\eval}{\mathsf{Eval}}

\newcommand{\sk}{\mathsf{SK}}

\newcommand{\tr}{\mathsf{Tr}}

\usepackage{cleveref}

\ignore{

}

\newcommand{\fclass}{\mathcal{F}}

\usepackage[english]{babel}
\usepackage{hyperref}
\addto\extrasenglish{

}

\newcommand{\cp}{\mathsf{CopyProtect}}
\newcommand{\uniform}{\xleftarrow{\$}}
\newcommand{\given}{\mid}
\newcommand{\ue}{\mathsf{UE}}
\newcommand{\vardbtilde}[1]{\tilde{\raisebox{0pt}[0.85\height]{$\tilde{#1}$}}}
\newcommand{\pubkue}{\mathsf{PBKUE}}
\newcommand{\ske}{\mathsf{SKE}}
\newcommand{\alice}{\mathcal{A}}
\newcommand{\bob}{\mathcal{B}}
\newcommand{\charlie}{\mathcal{C}}
\newcommand{\abc}{(\alice,\bob,\charlie)}
\newcommand{\abcprime}{(\alice',\bob',\charlie')}
\newcommand{\tldabc}{(\vardbtilde{\alice}, \vardbtilde{\bob}, \vardbtilde{\charlie})}
\newcommand\numberthis{\addtocounter{equation}{1}\tag{\theequation}}
\newcommand{\id}{\mathbf{id}}

\newcommand{\epr}[1]{\mathsf{EPR}_{#1}}
\newcommand{\E}{\mathbb{E}}
\newcommand{\forger}{\mathsf{Forger}}
\newcommand{\trD}[2]{T(#1, #2)}

\usepackage{breakcites}

\title{Unclonable Encryption, Revisited}
\author{Prabhanjan Ananth\\ UCSB\\ \texttt{prabhanjan@cs.ucsb.edu} \and Fatih Kaleoglu\\ UCSB\\ \texttt{kaleoglu@ucsb.edu}}
\date{}

\begin{document}

\maketitle

\begin{abstract}
\noindent Unclonable encryption, introduced by Broadbent and Lord (TQC'20), is an encryption scheme with the following attractive feature: given a ciphertext, an adversary cannot create two ciphertexts both of which decrypt to the same message as the original ciphertext. 
\par We revisit this notion and show the following: 
\begin{enumerate} 

\item {\bf Reusability}: The constructions proposed by Broadbent and Lord have the disadvantage that they either guarantee one-time security (that is, the encryption key can only be used once to encrypt the message) in the plain model or they guaranteed security in the random oracle model. We construct unclonable encryption schemes with semantic security. We present two constructions from minimal cryptographic assumptions: (i) a private-key unclonable encryption scheme assuming post-quantum one-way functions and, (ii) a public-key unclonable encryption scheme assuming a post-quantum public-key encryption scheme. 

\item {\bf Lower Bound and Generalized Construction}: We revisit the information-theoretic one-time secure construction of Broadbent and Lord. The success probability of the adversary in their construction was guaranteed to be $0.85^n$, where $n$ is the length of the message. It was interesting to understand whether the ideal success probability of (negligibly close to) $0.5^n$ was unattainable. We generalize their construction to be based on a broader class of monogamy of entanglement games (while their construction was based on BB84 game). We demonstrate a simple cloning attack that succeeds with probability $0.71^n$ against a class of schemes including that of Broadbent and Lord. We also present a $0.75^n$ cloning attack exclusively against their scheme.

\item {\bf Implication to Copy-Protection}: We show that unclonable encryption, satisfying a stronger property, called unclonable-indistinguishability (defined by Broadbent and Lord), implies copy-protection for a simple class of unlearnable functions. While we currently don't have encryption schemes satisfying this stronger property, this implication demonstrates a new path to construct copy-protection. 

\end{enumerate}

\end{abstract}

\newpage
\tableofcontents
\newpage

\section{Introduction}
Quantum mechanics has led to the discovery of many fascinating cryptographic primitives~\cite{wiesner83,Aar09,BGS13,Zhandry19,AGKZ20,BI20,GZ20,ALP20c,ALLZZ20} that are simply not feasible using classical computing. A couple of popular primitives include quantum money~\cite{wiesner83} and quantum copy-protection~\cite{Aar09}. We study one such primitive in this work. \par Inspired by the work of Gottesman~\cite{Got02} on tamper detection, Broadbent and Lord introduced the beautiful notion of unclonable encryption~\cite{BL19}. This notion is an encryption scheme that has the following attractive feature: given any encryption of a classical message $m \in \{0,1\}^*$, modeled as a quantum state, the adversary should be unable to generate multiple ciphertexts that encrypt to the same message. Formally speaking, the unclonability property is modeled as a game between the challenger and the adversary. The adversary consists of three algorithms, denoted by Alice, Bob and Charlie. The challenger samples a message $m$ uniformly at random and then sends the encryption of $m$ to Alice, who then outputs a bipartite state. Bob gets a part of this state and Charlie gets a different part of the state. Then the reveal phase is executed: Bob and Charlie each independently receive the decryption key. Bob and Charlie -- who no longer can communicate with each other -- now are expected to guess the message $m$ simultaneously. If they do, we declare that the adversary succeeds in this game. An encryption scheme satisfies unclonability property if any adversary succeeds in this game with probability at most negligible in the length of $m$. Note that the no-cloning principle~\cite{WZ82} of quantum mechanics is baked into this definition since if it were possible to copy the ciphertext, Alice can send this ciphertext to both Bob and Charlie who can then decrypt this using the decryption key (obtained during the reveal phase) to obtain the message $m$.  
\par Broadbent and Lord proposed two novel constructions of unclonable encryption. The drawback of their information-theoretic scheme is that it only guaranteed one-time security. This means that the encryption key can only be used to encrypt one message, after which the key can no longer be used to encrypt messages without compromising on security. On the other hand, their second scheme does provides reusable security, albeit only in the stronger random oracle model. Another (related) drawback is that their schemes were inherently private-key schemes, meaning that only the entity possessing the private encryption key could compute the ciphertext. 

\subsection{Our Work}

\paragraph{Reusability.} We revisit the notion of unclonable encryption of~\cite{BL19} and present two constructions. Both of our constructions guarantee semantic security; no information about the message is leaked even if the key is reused. The first construction is a private-key scheme (the encryption key is private) while the second construction is a public-key scheme (the encryption key is available to everyone). 

\begin{theorem}[Informal]
Assuming post-quantum one-way functions\footnote{A function $f$ is one-way and post-quantum secure if given $f(x)$, where $x \in \{0,1\}^{\secparam}$ is sampled uniformly at random, a quantum polynomial-time (QPT) adversary can recover a pre-image of $f(x)$ with probability only negligible in $\secparam$.}, there exists a private-key unclonable encryption scheme. 
\end{theorem} 

\begin{theorem}[Informal]
Assuming the existence of post-quantum public-key encryption schemes\footnote{An encryption scheme is said to be a post-quantum public-key encryption scheme if any quantum polynomial-time (QPT) adversary can distinguish encryptions of two equal-length messages $m_0,m_1$ with only negligible probability.}, there exists a public-key unclonable encryption scheme. 
\end{theorem}
\noindent We clarify that we show reusability only against distinguishing attacks and not cloning attacks. That is, the cloning attacker gets as input one ciphertext and in particular, does not get access to an encryption oracle. However, we do note that in the public-key setting, we can assume that the cloning adversary does not get access to the encryption oracle without loss of generality. Although in the private-key setting, a more delicate argument and/or construction is required.
\par Our constructions only guarantee computational security, unlike the previous scheme of Broadbent and Lord. However, our assumptions are the best one can hope for: (a) a private-key {\em unclonable} encryption scheme implies a post-quantum private encryption scheme (and thus, post-quantum one-way functions) and, (b) a public-key {\em unclonable} encryption scheme implies a public-key encryption scheme. There are candidates from lattices for both post-quantum one-way functions and post-quantum public-key encryption schemes; for example, see~\cite{Reg05}. 

\paragraph{Lower Bound and Generalized Construction.}
\noindent The first construction of~\cite{BL19}, \textit{conjugate encryption}, is based on the BB84 monogamy of entanglement game~\cite{Tomamichel_2013}, whose adversarial success probability is $\left( \frac{1}{2} + \frac{1}{2\sqrt{2}} \right) \approx 0.85$. In the hope of improving the bound, we present a simple generalization of their construction by showing a transformation from a broader class of monogamy games to unclonable encryption; whereas,~\cite{BL19} only showed the transformation for the BB84 monogamy game. \\

\noindent The optimal cloning adversary in conjugate encryption succeeds with probability $0.85^n$, whereas the ideal value would be negligibly close to $0.5^n$, where $n$ is the length of the messages, which is attainable trivially without cloning. A natural question to ask is if we can present a different analysis of their construction that gives the optimal bound. We show, in the theorem below, that this is not the case. 

\begin{theorem}[Informal]
    In a generalized conjugate encryption scheme which encrypts every bit of the message independently, a cloning adversary can succeed with probability at least $0.71^n$. 
\end{theorem}
\noindent The adversary that achieves this bound is simple: Alice clones the ciphertext with high fidelity using a generic cloning channel \cite{Bru__2000}. After learning the key, Bob and Charlie both try to honestly decrypt their state, and the output of the decryption matches the original message with significant probability for both of them. 
\par This adversarial construction inherently relies on the fact that the ciphertext (in qubits) is not larger than the message (in bits). For unclonable encryption schemes with large ciphertext size, it is infeasible to achieve a nontrivial bound using this technique. \\
\par The lower bound can be improved for conjugate encryption specifically using an adversary which blindly guesses part of the key before the splitting phase:

\begin{theorem}[Informal]
In conjugate encryption scheme of \cite{BL19}, a cloning adversary can succeed with probability $0.75^n$.
\end{theorem}

\paragraph{Implication to Copy-Protection.} We show how to use unclonable encryption to build quantum copy-protection~\cite{Aar09}. Roughly speaking, using a quantum copy-protection scheme, we can copy-protect our programs in such a way that an adversarial entity cannot create multiple versions of this copy-protected program. Recently, this notion has been revisited by many recent works~\cite{ALP20c,CMP20,ALLZZ20,KNY20,broadbent2021secure}. 
\par However, despite the recent progress, to date, we don't know of any provably secure constructions of copy-protection. We show how to use unclonable encryption to construct copy-protection for a specific class of point functions. This class consists of functions of the form $f_{a,b}(\cdot)$, where $b$ is a concatenation of the verification key and a signature on 0, that take as input $x$ and output $b$ if and only if $x=a$. This would not immediately yield a provably construction of copy-protection since we need the underlying unclonable encryption to satisfy a stronger property called {\em unclonable indistinguishability}  property (see \autoref{def:unclonableindistinguishable}) that are not currently satisfied by existing constructions of unclonable encryption. Nonetheless, this gives a new pathway to demonstrating provably secure constructions of quantum copy-protection.  

\begin{theorem}[Informal]
Assuming the existence of unclonable encryption scheme satisfying unclonable indistinguishability property and post-quantum one-way functions, there exists a quantum copy-protection scheme, satisfying computational correctness, for a special class of point functions. 
\end{theorem}

\noindent The resulting copy-protection guarantees a weaker correctness property called computational correctness property; informally, this says that any quantum polynomial-time adversary cannot come up with an input such that the copy-protected circuit is incorrect on this input. We note that such a correctness notion has been studied previously in the context of obfuscation~\cite{BLMZ19} (under the name computational functionality preservation). In addition to unclonable encryption, we use a post-quantum digital signature scheme that can be based on post-quantum one-way functions. 
 \par Our construction is inspired by a construction of secure software leasing by Broadbent, Jeffery, Lord, Podder and Sundaram~\cite{broadbent2021secure}. Conceptually, we follow the same approach suggested in their paper, except we replace the tool of quantum authentication codes \cite{BarnumetalQMAC02} with unclonable encryption. 
 \par Coladangelo, Majenz, and Poremba \cite{CMP20} also explore constructing copy-protection from unclonable encryption. They construct copy-protection for compute-and-compare programs\footnote{A compute and compare program implements a function $\mathsf{CC}[f,y]$ defined as:
\[ \mathsf{CC}[f,y](x) := \begin{cases} 1, \quad f(x) = y \\ 0, \quad f(x) \ne y \end{cases}. \] Point functions can be considered a special case when $f$ is the identity function.} (which subsumes point functions) in the QROM. Hence, whether unclonable-indistinguishable secure encryption can be constructed in the standard model is the key question in evaluating our contribution.  


\paragraph{Concurrent Works.} 
The work of Majenz, Schaffner and Tahmasbi~\cite{cryptoeprint:2021:408} study various limitations on unclonable encryption schemes. Specifically, they analyze lower bounds for the success probability of the adversary in any unclonable encryption scheme. In contrast, our lower bound targets specifically the conjugate encryption scheme of~\cite{BL19} and this allowed to present concrete lower bounds. 
\par Hiroka, Morimae, Nishimaki and Yamakawa~\cite{hiroka2021quantum} show how to make the key reusable for a different primitive called 
quantum encryption with certified deletion~\cite{BI20} using the same conceptual idea but different tools. We note that unclonable encryption implies quantum encryption with certified deletion if the certificate of deletion is allowed to be quantum. However, Hiroka et al.'s result achieves classical certification of deletion, which makes our results are incomparable.

\paragraph{Acknowledgements}
We thank the TCC 2021 committee for pointing out a simpler cloning attack against conjugate encryption.


\newcommand{\otue}{\mathsf{otUE}}
\newcommand{\rue}{\mathsf{rUE}}
\newcommand{\reuseenc}{\mathcal{E}}
\subsection{Technical Overview}
\noindent We present a high level overview of our techniques. 

\paragraph{Naive Attempt: A Hybrid Approach.} A naive attempt to construct an unclonable encryption scheme with reusable security is to start with two encryption schemes. 
\begin{itemize}
\item The first scheme is a (one-time) unclonable encryption scheme, as considered in the work of~\cite{BL19}. We denote this scheme by $\otue$.
\item The second scheme is a post-quantum encryption scheme guaranteeing reusable security but without any unclonability guarantees\footnote{As an example, we could use Regev's public-key encryption scheme~\cite{Reg05}.}. We denote this scheme by $\reuseenc$. 
\end{itemize}
\noindent At a high level, we hope that we can combine the above two schemes to get the best of both worlds: reusability and unclonability. 
\par In more detail, using $\otue$ and $\reuseenc$, we construct a reusable unclonable encryption scheme, denoted by $\rue$, as follows. Sample a decryption key $k_{\reuseenc}$ according to the scheme $\reuseenc$ and set the decryption key of $\rue$ to be $k_{\reuseenc}$. The encryption procedure of $\rue$ is defined as follows. To encrypt a message $m$, first sample a key $k_{\otue}$ according to the scheme $\otue$. Output the $\rue$ encryption of $m$ to be $(\ct_{\otue},\ct_{\reuseenc})$, where $\ct_{\otue}$ is an encryption of $m$ under the key $k_{\otue}$ and, $\ct_{\reuseenc}$ is an encryption of the message $k_{\otue}$ under the key $k_{\reuseenc}$. To decrypt, first decrypt $\ct_{\reuseenc}$ using $k_{\reuseenc}$ to obtain the message $k_{\otue}$. Using this, then decrypt $\ct_{\otue}$ to get the message $m$. 
\par How do we argue unclonability? Ideally, we would like to reduce the unclonability property of $\rue$ to the unclonability property of the underlying one-time scheme $\otue$. However, we cannot immediately perform this reduction. The reason being that $k_{\otue}$ is still encrypted under the scheme $\reuseenc$ and thus, we need to get rid of this key before invoking the unclonability property of $\otue$. To get rid of this key, we need to invoke the semantic security of $\reuseenc$. Unfortunately, we cannot invoke the semantic security of $\reuseenc$ since the decryption key of $\reuseenc$ will be revealed to the adversary and semantic security is trivially violated if the adversary gets the decryption key.
\par More concretely, Alice upon receiving $(\ct_{\otue},\ct_{\reuseenc})$ could first break $\ct_{\reuseenc}$ to recover $k_{\otue}$ and then decrypt $\ct_{\otue}$ using $k_{\otue}$ to recover $m$. Thus, before performing the reduction to $\rue$, we need to first invoke the security property of $\reuseenc$. Here is where we are stuck: as part of the security experiment of the unclonability property, we need to reveal the decryption key of $\rue$, which is nothing but $k_{\reuseenc}$, to Bob and Charlie after Alice produces the bipartite state. But if we reveal $k_{\reuseenc}$, then the security of $\reuseenc$ is no longer guaranteed.   

\paragraph{Embedding Messages into Keys.} To overcome the above issue, we require $\reuseenc$ to satisfy an additional property. Intuitively, this property guarantees the existence of an algorithm that produces a fake decryption key that has embedded inside it a message $m$ such that this fake decryption key along with an encryption of 0 should be indistinguishable from an honestly generated decryption key along with an encryption of $m$. 
\begin{quote}
    {\bf Fake-Key Property}: there is a polynomial-time algorithm $FakeGen$ that given an encryption of 0, denoted by $\ct_0$, and a message $m$, outputs a fake key $fk$ such that the distributions $\{(\ct_m,k_{\pke})\}$ and $\{(\ct_0,fk)\}$ are computationally indistinguishable, where $\ct_m$ is an encryption of $m$ and $k_{\pke}$ is the decryption key of $\pke$. 
\end{quote}
\noindent One consequence of the above property is that the decryption of $\ct_0$ using the fake decryption key $fk$ yields the message $m$. 
\par Using the above fake-key property, we can now fix the issue in the above hybrid approach. Instead of invoking semantic security of $\reuseenc$, we instead invoke the fake-key property of $\pke$. The idea is to remove $k_{\otue}$ completely in the generation $\ct_{\reuseenc}$ and only use it during the reveal phase, when the decryption key is revealed to both Bob and Charlie. That is, $\ct_{\reuseenc}$ is computed to be an encryption of 0 and instead of revealing the honestly generated key $k_{\reuseenc}$ to Bob and Charlie, we instead reveal a fake key that has embedded inside it the message $k_{\otue}$. After this change, we will now be ready to invoke the unclonability property of the underlying one-time scheme.

\paragraph{Instantiation: Private-Key Scheme.} We used a reusable encryption scheme $\reuseenc$ satisfying the fake-key property to construct an unclonable encryption satisfying reusable security. But does a scheme satisfying fake-key property even exist? 
\par We present two constructions: a private-key and a public-key encryption scheme satisfying fake-key property. We first start with a private-key encryption scheme. We remark that a slight modification of the classical private-key encryption scheme using pseudorandom functions~\cite{GolVol1} already satisfies this property\footnote{For the informed reader, this scheme can be viewed as a special case of a primitive called somewhere equivocal encryption~\cite{HOWW16}, considered in a completely different context.}. The encryption of a message $m$ using the decryption key $k_{\reuseenc}=(k,otp)$ is $\ct=(r,PRF_k(r) \oplus m \oplus otp)$, where $r \in \{0,1\}^{\secparam}$ is chosen uniformly at random, $\secparam$ is a security parameter and $PRF$ is a pseudorandom function. To decrypt a ciphertext $(r,\theta)$, first compute $PRF_k(r)$ and then compute $\theta \oplus PRF_k(r) \oplus otp$.
\par The fake key generation algorithm on input a ciphertext $\ct=(r,\theta)$ and a message $m$, generates a fake key $fk$ as follows: it first samples a key $k'$ uniformly at random and then sets $otp'$ to be $\theta \oplus PRF_{k'}(r) \oplus m$. It sets $fk$ to be $(k',otp')$. Note that $fk$ is set up in such a way that decrypting $\ct$ using $fk$ yields the message $m$. 

\paragraph{Instantiation: Public-Key Scheme.} We can present a construction of a public-key scheme using functional encryption~\cite{BSW11,ON10}, a fundamental notion in cryptography. A functional encryption (FE) scheme is an encryption scheme where the authority holding the decryption key (also referred to as master secret key) is given the ability to issue functional keys, of the form $sk_f$ for a function $f$, such that decrypting an encryption of $x$ using $sk_f$ yields the output $f(x)$. 
\par A first attempt to achieve fake-key property using FE is to design the fake key to be a functional key associated with a function, that has the message $m$, hardwired inside it. This function is a constant function that always ignores the input and outputs $m$. There are two issues with this approach: firstly, the fake key is a functional key whereas the real key is the master secret key of the functional encryption scheme. An adversary might be able to tell apart the fake key versus the real key and thus, break the security. Secondly, a public-key functional encryption does not guarantee function-hiding property -- the function description could be evident from the description of the functional key. This means that the adversary can read off the message $m$ from the description of the functional key. 
\par The first issue can be solved by making sure that even the real key is a functional key associated with the identity function. The second issue involves a little more work: instead of having $m$ in the clear in the description of the function, we instead hardwire encryption of $m$ in the function description. The decryption key for this ciphertext is encrypted inside the ciphertext of the FE scheme. Thus, we have two modes: (a) in the first mode, we encrypt $m$ using FE and the real key is a functional key associated with the identity function (this function has a dummy ciphertext hardwired inside it) and, (b) in the second mode, we encrypt $\widehat{k}$ using FE and the fake key is a functional key associated with a function, which has a ciphertext $c$ encrypting message $m$ hardwired inside it, that decrypts $c$ using $\widehat{k}$ and outputs the result. This trick is not new and is inspired by the Trojan technique~\cite{ABSV15} introduced in a completely different context. 
\par In the technical sections, instead of presenting a public-key encryption satisfying fake-key property using FE, we present a direct construction of public-key unclonable encryption scheme using FE.

\paragraph{Implication to Copy-Protection.} Next, we will show how to construct copy-protection for a specific class of point functions from unclonable encryption. A point function $f_{a,b}(\cdot)$ is represented as follows: it takes as input $x$ and outputs $b$ if $x=a$, otherwise it outputs $0$. Our approach is inspired by a recent work by Broadbent et al.~\cite{broadbent2021secure} who show how to construct a weaker version of copy-protection (called secure software leasing~\cite{ALP20c}) from quantum authentication codes. 
\par A first attempt to construct copy-protection, using unclonable encryption, is as follows: to copy-protect $f_{a,b}(\cdot)$, output an unclonable encryption\footnote{It suffices to use a one-time unclonable encryption scheme~\cite{BL19} here.} of $b$ under the key $a$; that is, $a$ is interpreted as the decryption key of the unclonable encryption scheme. We treat the ciphertext as the copy-protected version of $f_{a,b}(\cdot)$. To evaluate this copy-protected state on input $x$, run the decryption of this ciphertext with the key $x$. Output the result of the decryption algorithm.
\par If the input is $x=a$ then, by the correctness of unclonable encryption, we get the output $b$. However, if the input is not $a$, then we need the guarantee that the output is $0$ with high probability. Unfortunately, the properties of unclonable encryption fall short here. unclonable encryption does not have any guarantees if the ciphertext is decrypted using an invalid key. It could very well be the case that on input $a' \neq a$, the output of the copy-protection algorithm is $b$, thus violating the correctness guarantees. 
\par We use digital signatures to enforce the correctness property of the copy-protection scheme. We restrict our attention to a sub-class of point functions, where we interpret $b$ to be the concatenation of a verification key $vk$ and a signature $\sigma$ on 0. We subsequently modify the evaluation algorithm of the copy-protection scheme to output $(vk,\sigma')$ if and only if the decryption algorithm of unclonable encryption yields $(vk,\sigma')$ and moreover, $\sigma'$ is a valid signature on 0. This still does not guarantee the fact that the copy-protection scheme satisfies correctness. The reason being that on an input $a' \neq a$, the output could still be a valid signature on 0. Fortunately, this satisfies a weaker but still useful notion of correctness called computational correctness. This property states that an efficient adversary should not be able to find an input such that the evaluation algorithm outputs the incorrect value on this input. The reason why computational correctness holds is because it would be infeasible for the adversary to find an input such that the program outputs a valid signature on 0; if it did then it violates the unforgeability property of the underlying signature scheme.   
\par We need to show that given the copy-protected program, say $\rho$, an adversary cannot output two copies, say $\rho_1$ and $\rho_2$\footnote{Technically, this is incorrect since the two copies could be entangled and as written here, $\rho_1$ and $\rho_2$ are unentangled. But this is done just for ease of presentation, our argument can be suitably adapted to the general case.}, such that both of them evaluate $f_{a,b}(\cdot)$ with non-negligible probability. We prove this by contradiction. To show this, we first observe that we can get rid of the signature in the unclonable encryption ciphertext, by invoking the unclonable-indistinguishability property of the unclonable encryption scheme. This is where we crucially use the stronger indistinguishability property; this property allows us to change from one message to another message of our choice whereas in the (weaker) unclonability security property, the challenger is the one choosing the message to be encrypted. 
\par Now, we argue that there has to be a copy, say $\rho_1$ and evaluation algorithm $E_1$ (note that the adversary can choose the evaluation algorithms of its choice), such that when $E_1$ evaluates $\rho_1$ on the input $k$, where $k$ is the UE key, then we get a valid signature $\sigma$ on 0 with non-negligible probability.  Using $\rho_1$ we can then construct a forger that violates the unforgeability property of the digital signature scheme. 

\subsection{Structure of this Paper}
In \autoref{sec:prelims}, we give preliminary background and definitions.
In \autoref{sec:definitions}, we introduce natural definitions for many-time secure unclonable encryption in both private-key and public-key settings, as well as discuss the previous constructions given in \cite{BL19}. We give a construction for the private-key setting in \autoref{sec:privatekeyue} and for the public-key setting in \autoref{sec:publickeyue}. In \autoref{sec:bounds}, we present a generalized unclonable encryption construction using monogamy games, and a lower bound for conjugate encryption. \autoref{sec:copyprotection} shows that an unclonable encryption scheme satisfying unclonable-indistinguishable security (\autoref{def:unclonableindistinguishable}) implies copy-protection.

\newcommand{\mpk}{\mathsf{mpk}}
\newcommand{\Sim}{\mathsf{Sim}}
\newcommand{\queryset}{\mathsf{QSet}}
\newcommand{\values}{\mathcal{V}}
\newcommand{\density}[1]{\mathcal{D}\left( \mathcal{H}_{#1} \right)}
\section{Preliminaries} \label{sec:prelims}

\subsection{Notation}
We denote the security parameter by $\secparam$. We denote by $\negl(.)$ an arbitrary negligible function and by $\poly(.)$ an arbitrary function upper-bounded by a polynomial. We abbreviate probabilistic (resp., quantum) polynomial time by PPT (resp., QPT). 
\paragraph{} We denote by $\mathcal{M}$, $\mathcal{K}$, and $\mathcal{CT}$ (or  $\mathcal{H}_{\mathcal{CT}}$) the message space, the key space, and the ciphertext space, respectively. The message and the key are classical throughout this work, whereas the ciphertext can be classical or quantum, depending on the context. We sometimes use $0$ to denote a string of zeroes, the length of which will be clear from the context.

\subsection{Quantum Computing} Valid quantum states on a register $X$ are represented by the set of density operators on the Hilbert space $\mathcal{H}_X$, denoted by  $\density{X}$. A density operator $\rho: \mathcal{H}_X \to \mathcal{H}_X$ is defined a linear, positive semi-definite operator with unit trace, i.e. $\tr(\rho) = 1$, where $\tr$ is the trace operator. Density operators represent mixed quantum states, and a pure state $\ket{\psi} \in \mathcal{H}_X$ is represented by $\ket{\psi}\bra{\psi} \in \density{X}$.
\par Valid quantum operations from register $X$ to register $Y$ are represented by linear, completely positive trace-preserving (CPTP) maps $\phi: \density{X} \to \density{Y}$, also known as quantum channels. Valid quantum measurements on register $X$ with outcomes $x \in \mathcal{X}$ are represented by a positive operator-valued measure (POVM) on $\density{X}$, which is denoted by $F = (F_x)_{x\in X}$, where $F_x$ are positive semi-definite operators satisfying $\sum_x{F_x} = \id_X$, with $\id_X$ being the identity operator on $\mathcal{H}_\mathcal{X}$. The probability of measuring outcome $x$ on state $\rho$ equals $\tr(F_x \rho)$.
\par An EPR pair over $n$ qubits is a fully entangled bipartite $2n$-qubit state, defined as \begin{align*}
    \ket{\epr{n}} = \frac{1}{\sqrt{2^n}} \sum_{x \in \{0,1\}^n} \ket{xx},
\end{align*}
where $\left( \ket{x} \right)_{x \in \{0,1\}^n}$ is the standard basis.

\paragraph{Indistinguishability.} We define two distributions $\distr_0$ and $\distr_1$ to be computationally indistinguishable, denoted by $\distr_0 \approx_c \distr_1$, if any QPT distinguisher cannot distinguish the distributions $\distr_0$ and $\distr_1$.

\paragraph{Distance Measures.} There are two common distance measures considered in the literature: trace distance and fidelity. The fidelity of two quantum states $\rho,\sigma \in \density{A}$ is a measure of similarity between $\rho$ and $\sigma$ which is defined as 
$$F(\rho,\sigma) = \left(\tr\left(\sqrt{ \sqrt{\rho} \sigma \sqrt{\rho} }\right)\right)^2.$$
If $\sigma = \ket{\psi}\bra{\psi}$ is a pure state, the fidelity simplifies to $F(\sigma, \rho) = \braket{\psi}{\rho | \psi}$. We use the following useful fact: fidelity of two states does not increase under quantum operations. We state this fact from \cite{1996quant.ph..6012N} as a lemma below:
\begin{lemma}[Monotonicity of Fidelity] \label{lem:fidelity}
    Let $\rho, \sigma \in \density{X}$ and $\varphi: \density{X} \to \density{Y}$ be a CPTP map. Then, $$ F(\varphi(\rho), \varphi(\sigma)) \ge F(\rho, \sigma). $$
\end{lemma}

\noindent The trace distance of two states $\rho$ and $\sigma$, denoted by $T(\rho,\sigma)$ is defined as follows: 
$$T(\rho,\sigma) = \frac{1}{2}||\rho - \sigma ||_{tr} = \frac{1}{2} \tr\left( \sqrt{(\rho - \sigma)^{\dagger}(\rho - \sigma)}\right).$$

\paragraph{Almost As Good As New Lemma.} We use the Almost As Good As New Lemma\footnote{This is also known as the Gentle Measurement Lemma in the quantum information theory literature \cite{Winter_1999}.} \cite{aaronson2004limitations}, restated here verbatim from \cite{aaronson2016complexity}.

\begin{lemma}[Almost As Good As New]\label{lem:aagan} Let $\rho$ be a mixed state acting on $\mathbb{C}^d$. Let $U$ be a unitary and $(\Pi_0,\Pi_1=1-\Pi_0)$ be projectors all acting on $\mathbb{C}^d \otimes \mathbb{C}^d$. We interpret $(U,\Pi_0,\Pi_1)$ as a measurement performed by appending an ancillary system of dimension $d'$ in the state $\ket{0}\bra{0}$, applying $U$ and then performing the projective measurement $\{\Pi_0,\Pi_1\}$ on the larger system. Assuming that the outcome corresponding to $\Pi_0$ has probability $1-\varepsilon$, i.e., $\tr[\Pi_0(U\rho \otimes \ket{0}\bra{0}U^\dagger)]=1-\varepsilon$, we have $$\trD{\rho}{\widetilde{\rho}} \leq \frac{\sqrt{\varepsilon}}{2} ,$$
where $\widetilde{\rho}$ is state after performing the measurement and then undoing the unitary $U$ and tracing out the ancillary system: $$\widetilde{\rho} = \tr_{d'}\left(U^\dagger \left( \Pi_0U \left( \rho \otimes \ket{0}\bra{0} \right)U^\dagger \Pi_0 + \Pi_1U \left( \rho \otimes \ket{0}\bra{0} \right)U^\dagger \Pi_1\right) U \right) $$
\end{lemma}

\begin{corollary}\label{cor:uncomputing}
Let $\mathcal{Q}$ be a QPT algorithm which takes as input a state $\rho \in \density{A}$ and outputs a classical string $x\in X$. Then, $\mathcal{Q}$ can be reimplemented as $\mathcal{\widetilde{Q}}$ which satisfies the following properties: \begin{itemize}
    \item On input $\rho \in \density{A}$, $\mathcal{\widetilde{Q}}$ outputs $\rho' \otimes \ket{x'}\bra{x'} \in \density{A} \otimes \density{X}$ such that $$ \Pr \left[ x = x_0 : x \leftarrow \mathcal{Q}(\rho) \right] = \Pr \left[x' = x_0: \rho' \otimes \ket{x'} \bra{x'} \leftarrow \mathcal{\widetilde{Q}}(\rho) \right] $$ for any $x_0 \in X$. 
    \item For any state $\rho_0 \in \density{A}$ and a string $x_0 \in X$ satisfying $$ \Pr \left[ x = x_0 : \ket{x}\bra{x} \leftarrow \mathcal{Q}(\rho_0) \right] \ge 1 - \epsilon, $$ it holds that
    $$ \Pr \left[ x' = x_0 \land \trD{\rho'}{\rho_0} \le O(\sqrt{\epsilon}) : \rho' \otimes 
    \ket{x'} \bra{x'} \leftarrow \mathcal{\widetilde{Q}}(\rho_0) \right]. $$
\end{itemize}
In other words, $\mathcal{\widetilde{Q}}$ has the same functionality as $\mathcal{Q}$, and it also outputs a residual state $\rho'$ which is close to $\rho$ in trace distance provided that $\mathcal{Q}$ outputs the same string $x \in X$ probability close to 1 on input $\rho$.
\end{corollary}

\begin{proof}[Proof (sketch)] By the deferred measurement principle, we can transform $\mathcal{Q}$ into the following form without changing its functionality: \begin{itemize}
    \item It appends to $\rho$ an ancillary system initialized at $\ket{0}\bra{0}_B$.
    \item It applies a unitary $U$ to the bipartite state $\rho \otimes \ket{0}\bra{0}_B$ to obtain $\rho_{AB}$.
    \item It performs a POVM on $\rho_{AB}$ to measure $\ket{x}\bra{x} \in \density{X}$ and output $x$. Let the residual state be $\widetilde{\rho_{AB}}$ after this measurement.
\end{itemize}
$\mathcal{\widetilde{Q}}$ performs the steps above, and then recovers $\rho$ by applying $U^\dagger$ to $\widetilde{\rho_{AB}}$ and tracing out the ancillary system.
The analysis is essentially the same as that in \autoref{lem:aagan}, and we refer the reader to \cite{aaronson2016complexity} for details.
\end{proof}

\subsection{Post-Quantum Digital Signatures}
Post-quantum signature schemes with perfect correctness, defined below, can be constructed from post-quantum secure one-way functions:
\begin{definition}[Post-Quantum Signature Scheme] \label{def:signature}
A post-quantum signature scheme over a message space $\mathcal{M}$ is a tuple of PPT algorithms $(\gen, \sign, \ver)$: \begin{itemize}
    \item{\bf Key Generation:} $\gen(1^\secparam)$ takes as input a security parameter and outputs a pair of keys $(vk, sk)$.
    \item{\bf{Signing:}} $\sign(sk, m)$ takes as input the secret (signing) key $sk$ and a message $m \in \mathcal{M}$. It outputs a signature $\sigma$.
    \item {\bf Signature Verification:} $\ver(vk, m, \sigma')$ takes as input the verification key $vk$, a message $m \in \mathcal{M}$ and a candidate signature $\sigma'$. It outputs a bit $b\in \{0,1\}$.
\end{itemize}
which satisfy correctness and unforgeability properties defined below:
\begin{itemize}
    \item {\bf Correctness:} For all messages $m \in \mathcal{M}$, we have $$ \Pr[b = 1: (vk, sk) \leftarrow \gen(1^\secparam), \; \sigma \leftarrow \sign(sk, m), \; b \leftarrow \ver(vk, m, \sigma) ] = 1.$$
    \item {\bf Post-Quantum (One-Time) Existential Unforgeability:} For any QPT adversary $\alice$ and any message $m \in \mathcal{M}$, we have: $$ \Pr[1 \leftarrow \ver(vk, m, \sigma'): (vk, sk) \leftarrow \gen(1^\secparam), \; \ket{\sigma'}\bra{\sigma'} \leftarrow \alice(vk)] \le \negl(\secparam).$$
\end{itemize}
\end{definition}

\noindent Post-quantum digital signatures can be based on post-quantum one-way functions~\cite{Rompel90}. 

\subsection{Functional Encryption}
\label{sec:fe}
\noindent A functional encryption scheme allows a user to decrypt an encryption of a message $x$ using a functional key associated with $C$ to obtain the value $C(x)$. The security guarantee states that the user cannot learn anything beyond $C(x)$. Depending on the number of functional keys issued in the security experiment, we can consider different versions of functional encryption. Of interest to us is the notion of single-key functional encryption where the adversary can only query for a single functional key during the security experiment. 
\par A public-key functional encryption scheme $\fe$ associated with a class of boolean circuits $\cktclass$ is defined by the following algorithms.
\begin{itemize}

\item {\bf Setup,} $\setup(1^{\secparam},1^s)$: On input security parameter $\secparam$, maximum size of the circuits $s$ for which functional keys are issued, output the master secret key $\msk$ and the master public key $\mpk$.

\item {\bf Key Generation,} $\keygen(\msk,C)$: On input master secret key $\msk$ and a circuit $C \in \cktclass$ of size $s$, output the functional key $\sk_{C}$.

\item {\bf Encryption,} $\enc(\mpk,x)$: On input master public key $\mpk$, input $x$, output the ciphertext $\ct$.

\item {\bf Decryption,} $\dec(\sk_C,\ct)$: On input functional key $\sk_C$, ciphertext $\ct$, output the value $y$.

\end{itemize}

\begin{remark}
A private-key functional encryption scheme is defined similarly, except that $\setup(1^{\secparam},1^s)$ outputs only the master secret key $\msk$ and the encryption algorithm $\enc$ takes as input the master secret key $\msk$ and the message $x$.
\end{remark}

\noindent A functional encryption scheme satisfies the following properties.

\paragraph{Correctness.} Consider an input $x$ and a circuit $C \in \cktclass$ of size $s$. We require the following to hold for every $Q \geq 1$:
$$\prob \left[ C(x) \leftarrow \dec(\sk_C,\ct)\ :\ \substack{(\mpk,\msk) \leftarrow \setup(1^{\secparam},1^s);\\ \sk_C \leftarrow  \keygen(\msk,C); \\ \ct \leftarrow \enc(\mpk,x) } \right] \geq 1 - \negl(\secparam),$$
for some negligible function $\negl$.

\paragraph{Single-Key Security.} We only consider functional encryption schemes satisfying single-key security property. To define the security of a single-key functional encryption scheme $\fe$, we define two experiments $\expt_0$ and $\expt_1$. Experiment $\expt_0$, also referred to as {\em real} experiment, is parameterized by a PPT stateful adversary $\adversary$ and a challenger $\ch$. Experiment $\expt_1$, also referred to as the {\em simulated} experiment, is parameterized by a PPT adversary $\adversary$ and a PPT stateful simulator $\Sim$.  \\

\noindent \underline{$\expt_0^{\fe,\adversary,\ch}(1^{\secparam})$}:
\begin{itemize}

\item $\adversary$ outputs the maximum circuit size $s$.

\item $\ch$ executes $\fe.\setup(1^{\secparam},1^s)$ to obtain the master public key-master secret key pair $(\mpk,\msk)$. It sends $\mpk$ to $\adversary$. 

\item {\bf Challenge Message Query}: After receiving $\mpk$, $\adversary$ outputs the challenge message $x$. The challenger computes the challenge ciphertext $\ct \leftarrow \enc(\mpk,x)$. $\ch$ sends $\ct$ to $\adversary$. 

\item {\bf Circuit Query}: $\adversary$ upon receiving the ciphertext $\ct$ as input, outputs a circuit $C$ of size $s$. The challenger then sends $\sk_{C}$ to $\adversary$, where $\sk_{C} \leftarrow \keygen(\msk,C)$. 

\item Finally, $\adversary$ outputs the bit $b$.

\end{itemize}

\noindent \underline{$\expt_1^{\fe,\adversary,\Sim}(1^{\secparam})$}:
\begin{itemize}

\item $\adversary$ outputs the maximum circuit size $s$.

\item $\Sim$, on input $(1^{\secparam},1^s)$, outputs the master public key $\mpk$.

\item {\bf Challenge Message Query}: $\adversary$ upon receiving a public key $\mpk$, outputs a message $x$. $\Sim$, upon receiving $1^{|x|}$ (i.e., only the length of the input) as input, outputs the challenge ciphertext $\ct$.

\item {\bf Circuit Query}: $\adversary$ upon receiving the ciphertext $\ct$ as input, outputs a circuit $C$ of size $s$. $\Sim$ on input $(C,C(x))$, outputs a functional key $\sk_C$. 

\item Finally, $\adversary$ outputs a bit $b$.

\end{itemize}

\noindent A single-key public-key functional encryption scheme is secure if the output distributions of the above two experiments are computationally indistinguishable. More formally,

\begin{definition}
A single-key public-key functional encryption scheme $\fe$ is 	{\bf secure} if for every large enough security parameter $\secparam \in \mathbb{N}$, every PPT adversary $\adversary$, there exists a PPT simulator $\Sim$ such that the following holds:
$$\left| \prob \left[ 0 \leftarrow \expt_{0}^{\fe,\adversary,\ch}(1^{\secparam}) \right] - \prob \left[ 0 \leftarrow \expt_1^{\fe,\adversary,\Sim}(1^{\secparam}) \right] \right| \leq \negl(\secparam),$$
for some negligible function $\negl$.
\end{definition}

\paragraph{Instantiations.} A single-key public-key functional encryption scheme can be built from any public-key encryption scheme~\cite{SS10,GVW12}. If the underlying public-key encryption scheme is post-quantum secure then so is the resulting functional encryption scheme. 
\subsection{Quantum Copy-Protection}

\noindent Below we present the definition of a copy-protection scheme, adapted from \cite{broadbent2021secure} and originally due to \cite{Aar09}.
\begin{definition} [Copy-Protection Scheme] \label{def:copyprotection}
 Let $\fclass = \fclass(\secparam)=\{f:X \to Y\}$ be a class of efficiently computable functions functions. A copy protection scheme for $\fclass$ is a pair of quantum algorithms $(\cp, \eval)$ such that for some output space $\density{Z}$: \begin{itemize}
    \item \textbf{Copy Protected State Generation:} $\cp(1^\secparam, d_f)$ takes as input the security parameter $1^\secparam$ and a classical description $d_f$ of a function $f \in \fclass$ (that efficiently computes $f$). It outputs a mixed state $\rho_f \in \density{Z}$.
    \item \textbf{Evaluation:} $\eval(1^\secparam, \rho, x)$ takes as input the security parameter $1^\secparam$, a mixed state $\rho\in \density{Z}$, and an input value $x\in X$. It outputs a bipartite state $\rho' \otimes \ket{y}\bra{y} \in \density{Z} \otimes \density{Y}$. 
    
\end{itemize}
\end{definition}


\paragraph{Correctness:} Informally speaking, if an honestly generated copy-protected state $\rho_f$ for a function $f \in \fclass$ is honestly evaluated using $\eval$ on any input $x \in X$, the output should be $f(x)$. We defer the formal definition of correctness to \autoref{sec:pointfunctions}, where we define a weaker notion of computational correctness specifically for point functions, which is the context we use copy-protection in throughout \autoref{sec:copyprotection}.

\paragraph{Security.} Security in the context of copy-protection means that given a copy-protected program $\rho_f$ of a function $f \in \fclass$, no QPT adversary can produce two programs that can both be used to compute $f$. This is captured in the following definition adapted by the "malicious-malicious security" definition given in \cite{broadbent2021secure}:
\begin{definition}[Copy-Protection Security] \label{def:unclonability}
A copy-protection scheme $(\cp, \eval)$ for a class $\fclass$ of functions $f:X \to Y$ and a distribution $\distr$ over $\fclass$ is $\delta(\secparam)$-\textbf{secure} with respect to a family of distributions $\left\{\distr_X^f\right\}_{f \in \fclass}$ over $X$ if any QPT adversary $\abc$ cannot succeed in the following \textbf{pirating experiment} with probability greater than $\delta(\secparam) + \negl(\secparam)$: \begin{itemize}
    \item The challenger samples a function $f\leftarrow \distr$ and sends $\rho_f \leftarrow \cp(1^\secparam, d_f)$ to $\alice$.
    \item $\alice$ applies a CPTP map to split $\rho_f$ into a bipartite state $\rho_{BC}$, and sends the $B$ (resp., $C$) register to $\bob$ (resp., $\charlie$). No communication is allowed between $\bob$ and $\charlie$ after this step.
    \item The challenger samples $x \leftarrow \distr_X^f$ and sends $x$ to both $\bob$ and $\charlie$.
    \item $\bob$ (resp., $\charlie$) outputs\footnote{Since $\bob$ and $\charlie$ cannot communicate, the order in which they use their share of the copy-protected program is insignificant.} $y_B \in Y$ (resp., $y_C \in Y$). The adversary wins if $y_B = y_C = f(x)$.
\end{itemize}
Note that this definition is referred to as malicious-malicious security because the adversary is free to choose the registers $B,C$ as well as the evaluation algorithms used by $\bob$ and $\charlie$.
\end{definition}


\newcommand{\process}{\mathsf{Process}}
\subsection{Copy-Protection of Point Functions} \label{sec:pointfunctions}
\paragraph{Point Functions:} Let $a$ and $b$ be binary strings. The point function $f_{a,b}: \{0,1\}^{|a|} \to \{0,1\}^{|b|}$ is defined as $$ f_{a,b}(x) = \begin{cases} b, \quad x = a  \\ 0, \quad x \ne a \end{cases}.$$

\noindent Ordinarily, one would define the correctness property of a copy-protection scheme as follows: an honest evaluation of $f(x)$ using an honestly generated copy-protected state $\rho_f$ for $f$ succeeds with small error for all $x$. For point functions, we define a weaker notion of correctness, which states that it is computationally hard to find an input $x$ which fails honest evaluation. In a bit more detail, an adversary is given a copy-protected program for the point function $f_{a,b}$. Firstly, if he uses this program to honestly evaluate $f_{a,b}$ on input $a$, then he will obtain output $b$ and the program will not be destroyed. Secondly, if he does not have auxiliary information and he only uses $\eval()$ to query $f_{a,b}$, then he will not come across an input that evaluates incorrectly except with small probability. We formalize this second condition as a correctness experiment.


\begin{definition}[Computational Correctness] \label{def:computationalcorrectness}
A copy-protection scheme $(\cp, \eval)$ for a class of point functions $\fclass=\{f_{a,b}\ :\ (a,b) \in X \times Y \}$, where $X = \{0,1\}^{\poly(\secparam)}$ and $Y = \{0,1\}^{\poly(\secparam)}$, satisfies \textit{computational $(\varepsilon(\secparam),\delta(\secparam))$-correctness} with respect to a probability distribution $\distr$ over $\fclass$ if: \begin{itemize}
    \item For any $f_{a,b} \in \mathcal{F}$, we have: $$\Pr[\rho' \otimes \ket{b}\bra{b} \leftarrow \eval(1^{\secparam},\rho,a) \wedge T(
    \rho,\rho') \leq \varepsilon(\secparam): \rho \leftarrow \cp(1^\secparam, (a,b))] = 1,$$
    where $T(\cdot,\cdot)$ denotes trace distance.
    \item No QPT adversary $\alice$ can succeed in the following correctness experiment with probability greater than $\delta(\secparam)$: 
    \begin{itemize}
        \item The challenger samples $f_{a,b} \leftarrow \distr$ and computes $\rho_f^{(0)} \leftarrow \cp(1^\secparam, (a,b))$.
        \item For $i = 0,1,\dots,\poly(\secparam)$; $\alice$ sends an adaptive query $x_i \in X$ to the challenger, who computes $\rho_f^{(i+1)} \otimes \ket{y_i}\bra{y_i} \leftarrow \eval(1^\secparam, \rho_f^{(i)}, x_i)$ and sends $y_i \in Y$ back to $\alice$.
        \item $\alice$ wins if there exists an index $i \in \{0,1,\dots,\poly(\secparam)\}$ such that $x_i \ne a$ and $y_i \ne 0$.
    \end{itemize}
\end{itemize}
\end{definition}


\begin{remark} \label{rem:computationalcorrectness}
To give more context on this definition, imagine a scenario where a software firm (Alice) provides a copy-protected program $\rho_f$ to a client (Bob). Computational correctness guarantees that if the client follows the instructions provided by $\alice$, that is, if he only uses $\rho_f$ as an input to the algorithm $\eval()$, then he will get the correct output with overwhelming probability. This is true even if $\rho_f$ changes greatly after Bob evaluates the function $f$. However, $\rho_f$ has no reusability guarantee once Bob uses third party programs that modify $\rho_f$. Our definition is closely related to the notion of "computational functionality preservation" defined in \cite{BLMZ19} in the context of classical virtual-black-box obfuscation, which states given an obfuscated program, a PPT adversary cannot find an input which evaluates incorrectly. Note that the issue of the program being destroyed is specific to the quantum setting.

\end{remark}

\begin{remark}\label{rem:distributionalcorrectness}
Computational correctness is stronger than distributional correctness defined in \cite{broadbent2021secure}, which states that honest evaluation yields the correct output with probability close to 1, when the input is sampled from some distribution over the input space, as long as the distribution is efficiently samplable (in particular the uniform distribution). The reason is simple: a QPT adversary can sample the query input from this distribution.
\end{remark}

\begin{definition}[Copy-Protection Security for Point Functions] \label{def:copyprotectionpointfunction}
A copy-protection scheme $(\cp, \allowbreak \eval)$ for a class of point functions $\fclass=\{f_{a,b}\ :\ a \in X, b \in Y\}$ is called \textbf{secure} if it is $\frac{1}{2}$-secure with respect to $\left\{\distr_X^f\right\}_{f \in \fclass}$, where \begin{align*} &\Pr[x = a: x \leftarrow \distr_X^{f_{a,b}} ] = \frac{1}{2}, \\
&\Pr[x = a': x \leftarrow \distr_X^{f_{a,b}}] = \frac{1}{2|X| - 2}
\end{align*}
for all $f_{a,b} \in \fclass$ and $a' \ne a$. That is, $\distr_X^{f_{a,b}}$ samples $a$ with probability $1/2$ and every other $a' \ne a$ with equal probability.
\end{definition}

\par Note that an adversary can trivially succeed in the pirating experiment for point functions with probability $1/2$ by always outputting 0 in both registers.

\section{Private-Key and Public-Key Unclonable Encryption: Definition}
\label{sec:definitions}
We present the definitions of public-key and private-key unclonable encryptions, satisfying reusable security. Before we present these definitions, we first define an unclonable encryption scheme borrowed from \cite{BL19}.  

\subsection{Unclonable Encryption} 
\label{sec:otue}

\begin{definition}[Quantum Encryption of Classical Messages (QECM)]
\label{def:qecm}
A QECM scheme $\ue$ is a tuple of QPT algorithms $(\ue.\setup,\ue.\enc,\ue.\dec)$: 
\begin{itemize} 

\item {\bf Setup, $\ue.\setup(1^{\secparam})$}: on input the security parameter $\secparam$, it outputs a key $k \in \mathcal{K}$. 

\item {\bf Encryption, $\ue.\enc(k,m)$}: on input a the key $k$ and a message $m \in \mathcal{M}$, it outputs a ciphertext $\ct \in \mathcal{H_{CT}}$. 

\item {\bf Decryption, $\ue.\dec(k,\ct)$}: on input a key $k \in \mathcal{K}$ and a ciphertext $\ct \in \mathcal{H_{CT}}$, it outputs a message $m' \in \mathcal{M}$. 
\end{itemize}
A public-key QECM is defined analogously.
\end{definition}

\paragraph{Statistical Correctness:} For any key $k \in \mathcal{K}$ and any message $m \in \mathcal{M}$ we have $$ \Pr \left[ m' = m: \; \ket{\ct}\bra{\ct} \leftarrow \ue.\enc(k,m), \; \ket{m'}\bra{m'} \leftarrow \ue.\dec(k, \ct) \right] \ge 1 - \negl(\secparam). $$

We consider two types of security notions: indistinguishability and unclonability. The former states that encryption hides the message in the absence of any knowledge of the key.

\paragraph{Indistinguishable Security:}

\begin{definition}[(One-Time) Indistinguishable Security] \label{def:indistinguishable}
We say that a QECM $\otue = (\otue.\setup, \allowbreak \otue.\enc, \otue.\dec)$ is indistinguishable-secure if for any messages $m_1, m_2 \in \mathcal{M}$ of equal length, the following holds: \begin{align*}
    \left\{\otue.\enc(k,m_1)\  : \ k \leftarrow \otue.\setup(1^\secparam) \right\} \approx_c \left\{\otue.\enc(k,m_2)\ : \ k \leftarrow \otue.\setup(1^\secparam)\right\}.
\end{align*}

\end{definition}

\noindent If we allow the encryption key to be reusable, we arrive at the notion of many-time indistinguishability, also known as semantic security.

\begin{definition}[Semantic Security] \label{def:semanticsecurity}
A QECM $(\setup,\enc,\dec)$ is said to satisfy semantic security if it satisfies the following property: for sufficiently large $\secparam \in \mathbb{N}$, for every $(m_1^{(0)},\ldots,m_q^{(0)}),(m_1^{(1)},\allowbreak \ldots,m_q^{(1)})$ such that $|m_i^{(0)}|=|m_i^{(1)}|$ for every $i \in [q]$ and $q=\mathrm{ poly}(\secparam)$,
$$\left\{\enc \left(k,m_1^{(0)} \right),\ldots,\enc \left( k,m_q^{(0)} \right)\right\} \approx_c \left\{ \enc \left(k,m_1^{(1)} \right),\ldots,\enc \left(k,m_q^{(1)} \right) \right\},$$
where $k \leftarrow \setup(1^{\secparam})$. \\ \\

\end{definition}

\noindent Semantic security for public-key QECM is defined analogously.

\begin{definition}[Semantic Security for Public-Key QECM] \label{def:pubkeysemanticsecurity}
A public-key QECM $(\setup,\enc,\dec)$ is said to satisfy semantic security if the following holds: for sufficiently large $\secparam \in \mathbb{N}$, for every $m_0,m_1$ of equal length, 
$$\{\enc(\pk,m_0)\} \approx_c \{\enc(\pk,m_1)\},$$
the distinguisher also receives as input $\pk$, where $\pk$ is such that $(\pk,\sk) \leftarrow \setup(1^{\secparam})$. 
\end{definition}

\paragraph{Unclonable Security:} 
Unclonable security states that a ciphertext cannot be cloned while preserving its decryption functionality. 

\begin{definition}[Unclonable Security]
\label{def:unclonable}
We say that a QECM with message length $n$ is \textbf{$t$-unclonable secure} if a QPT cloning adversary $\abc$ cannot succeed with probability more than $2^{-n+t} + \negl(\secparam)$ in the cloning experiment defined below:

\paragraph{Cloning Experiment:}
The cloning experiment consists of two phases:
\begin{itemize}
    \item In phase 1, the challenger samples a key $k \leftarrow \setup(1^\secparam)$ and a message $m \in \mathcal{M}$ uniformly at random. He then computes $\rho_{\ct}$ and sends it to $\alice$, who applies to $\rho_\ct$ a CPTP map $\phi: \density{A} \to \density{B} \otimes \density{C}$ to obtain the bipartite state $\rho_{BC}$. She sends the $B$ (resp., C) register of this state to $\bob$ (resp., $\charlie$). 
    \item In phase 2, $\bob$ and $\charlie$ are not allowed to communicate. The key $k$ is revealed to both of them. Then, $\bob$ (resp., $\charlie$) applies a POVM $B^k$ (resp., POVM $C^k$) to their register to measure and output a message $m_B$ (resp., $m_C$).
    \item The adversary wins iff $m_B = m_C = m$.
    
\end{itemize}

\end{definition}

\begin{remark} \label{rem:onetimeonly}
In this work, we only consider one-time unclonability, meaning the adversary is tasked to create two ciphertexts out of one. A natural extension of this notion would be to require that an adversary cannot create $m + 1$ ciphertexts out of $m$.
\end{remark}

\noindent Below is a stronger notion of security which implies both \autoref{def:indistinguishable} and \autoref{def:unclonable}, which is called unclonable-indistinguishable security\footnote{We slightly deviate from \cite{BL19} in defining unclonable-indistinguishable security. We have the adversary choose two messages whereas they require one of the messages to be a uniformly random message. We anticipate that the two definitions may be equivalent.}.

\begin{definition}[Unclonable-Indistinguishable Security]
\label{def:unclonableindistinguishable}
We say that a QECM with message length $n$ is unclonable-indistinguishable secure if a QPT cloning-distinguishing adversary $\abc$ cannot succeed with probability more than $1/2 + \negl(\secparam)$ in the cloning-distinguishing experiment defined below:

\paragraph{Cloning-Distinguishing Experiment:}
The cloning experiment consists of two phases:
\begin{itemize}
    \item In phase 1, $\alice$ chooses two messages $m_0,m_1 \in \mathcal{M}$ and sends them to the challenger. The challenger samples a key $k \leftarrow \setup(1^\secparam)$ and a bit $b$ uniformly at random. The challenger then computes $\rho \leftarrow \enc(k, m_b)$ and sends $\rho$ to $\alice$.
    \item In phase 2, $\alice$ has a ciphertext $\rho$ to which she applies a CPTP map $\phi: \density{A} \to \density{B} \otimes \density{C}$ to split it into two registers ($B,C$). She then sends the $B$ and $C$ registers to $\bob$ and $\charlie$, respectively.
    \item In phase 3, the key $k$ is revealed to both $\bob$ and $\charlie$. Then, $\bob$ (resp., $\charlie$) applies a POVM $B^k$ (resp., POVM $C^k$) to their register to measure and output a bit $b_B$ (resp., $b_C$).
    \item The adversary wins iff $b_B = b_C = b$.
    
\end{itemize}

\end{definition}

\paragraph{Instantiations.} The work of Broadbent and Lord~\cite{BL19} presented two constructions of one-time unclonable encryption, that is, constructions satisfying \autoref{def:indistinguishable} and \autoref{def:unclonable}. Their first construction, "conjugate encryption", which encrypts messages of constant length $n$, is information-theoretic and $n\log_2(1 + 1/\sqrt{2})$-unclonable secure. This scheme upper-bounds the success probability of a cloning adversary by $1/2 + 1/2\sqrt{2} \approx 0.85$ in the single-bit message ($n=1$) case. \\ \\
\noindent The second construction, "$\mathcal{F}-$conjugate encryption", is based on computational assumptions. It uses post-quantum pseudo-random functions but is only shown to be secure in the random oracle model. Nonetheless, it satisfies multi-message security and $\log_2(9)$-unclonable security for long messages. Their analysis for this scheme does not provide an unclonability bound for the single-bit message case. \\ \\
\noindent There is no known construction of an unclonable-indistinguishable secure scheme that we know of. For instance, no qubit-wise encryption scheme, including conjugate-encryption and generalized conjugate encryption (\autoref{appendix:bounds1}) can satisfy that definition for messages of length $n \ge 2$, since the cloning-distinguishing adversary can send one half of the ciphertext to $\bob$ and the other half to $\charlie$.

\paragraph{Conjugate Encryption Upper and Lower Bounds} \cite{BL19} shows that in their conjugate encryption scheme a cloning adversary can succeed with probability at most $(1/2 + 1/2\sqrt{2})^n$, which is based on BB84 monogamy-of-entanglement (MOE) game analyzed in \cite{Tomamichel_2013}. Their proof technique can be generalized to a class of MOE games, which we call \textit{real-orthogonal monogamy games}, to potentially obtain better security in the event that a monogamy game with a better value exists in this class. \\

\noindent Arbitrary pure single-qubit states on the $xz$ plane of the Bloch Sphere can be cloned with fidelity $f := (1/2 + 1/2\sqrt{2}) \approx 0.85$ \cite{Bru__2000}. Since every ciphertext lies on the $xz$ plane in conjugate encryption, a cloning adversary (for each qubit) clone the ciphertext with fidelity $f$. In phase 2, both $\bob$ and $\charlie$ will decrypt their register, hence each having fidelity $f$ to the message $\ket{m}\bra{m}$. By union bound, this implies that they both output $m$ with probability at least $(2f-1)^n \approx 0.7^n$. In the single-bit message case, this means that the scheme of~\cite{BL19}, and a class of similar constructions, can be violated by an adversary with probability 0.7. Conjugate encryption in particular can be attacked with probability $0.75$ for single-bit messages. For details regarding these upper-lower bounds, see \autoref{appendix:bounds1} and \autoref{appendix:bounds2}.


\subsection{Private-Key and Public-key Unclonable Encryption}
\noindent Having established the preliminaries, we are ready to present the definitions of one-time unclonable encryption as well as reusable unclonable encryption in the private-key and public-key settings.

\begin{definition}
A QECM $\otue = (\otue.\setup, \otue.\enc, \otue.\dec)$ is called a \textbf{one-time unclonable encryption} scheme if it satisfies one-time indistinguishable security (\autoref{def:indistinguishable}) as well as unclonable security (\autoref{def:unclonable}).
\end{definition}

\begin{definition}
\label{def:privkue}
A QECM $(\setup,\enc,\dec)$ is called a \textbf{private-key unclonable encryption scheme} if it satisfies the properties of (reusable) semantic security (\autoref{def:semanticsecurity}) and unclonable security. 
\end{definition}

\begin{definition}
\label{def:pubue}
A public-key QECM $(\setup,\enc,\dec)$, is called a \textbf{public-key unclonable encryption scheme} if it satisfies public-key semantic security (\autoref{def:pubkeysemanticsecurity}) and unclonable security.
\end{definition}
\noindent For a construction of public-key encryption using functional encryption, see \autoref{sec:publickeyue}.

\newcommand{\fakegen}{\mathsf{FakeGen}}

\section{Private-Key Unclonable Encryption (PK-UE)}
\label{sec:privatekeyue}
\noindent We present a construction of (reusable) private-key unclonable encryption in this section. One of the tools required in our construction is a private-key encryption with fake-key property. We first define and construct this primitive. 

\subsection{Private-Key Encryption with Fake-Key Property}
\label{sec:privencfk} 
\noindent We augment the traditional notion of private-key encryption with a property, termed as fake-key property. 
This property allows an authority to issue a fake decryption key $fk$, as a function of $m$ along with an encryption of $m$, denoted by $\ct$, in such a way that a QPT distinguisher will not be able to distinguish whether it received the real decryption key or a fake decryption key. A consequence of this definition is that, the decryption algorithm on input the fake decryption key $fk$ and $\ct$ should yield the message $m$. 

\begin{definition}[Fake-Key Property]
We say that a classical encryption scheme $(\setup,\enc,\dec)$ satisfies the \textbf{fake-key property} if there exists a polynomial time algorithm $\fakegen: \mathcal{CT} \times \mathcal{M} \to \mathcal{K}$ such that for any $m \in \mathcal{M}$, \begin{align*} 
    \left\{\left(ct^m \leftarrow \enc(k,m),k \right)\right\} \approx_c \left\{\left(ct^0 \leftarrow \enc(k,0), \ fk \leftarrow \fakegen(ct^0, m) \right)\right\}, \numberthis \label{fakekeyproperty}
\end{align*}
where $k \leftarrow \setup(1^{\secparam})$.
\end{definition}

\noindent Note that in particular, the fake-key property requires that $\dec(fk, ct^0) = m$.

\begin{theorem}
Assuming the existence of post-quantum pseudorandom functions, there exists a classical private-key encryption scheme (PKE) that satisfies the fake-key property.
\end{theorem}

\begin{proof} Let $\{PRF_k: \{0,1\}^{\ell} \to \{0,1\}^n\ :\ k \in \{0,1\}^{\secparam}\}$ be a class of post-quantum pseudo-random functions, where $\ell$ is set to be $\secparam$ and $n$ is the length of the messages encrypted.
\par Consider the following scheme: \begin{itemize}
    \item \textbf{Setup, $\setup(1^{\secparam})$:} on input $\secparam$, it outputs $(k, otp)$, where $k \leftarrow \{0,1\}^\secparam$ and $otp \leftarrow \{0,1\}^n$ are uniformly sampled.
    
    \item \textbf{Encryption, $ \enc((k,otp),m)$:} on input key $(k,otp)$, message $m \in \{0,1\}^n$, it outputs $ct = (ct_1,ct_2)$, where $\ct_1 = r$ and $ct_2= PRF_k(r) \oplus m \oplus otp$ with $r \leftarrow \{0,1\}^\ell$ being uniformly sampled.
    
    \item \textbf{Decryption, $\dec((k,otp), ct)$:} on input $(k,otp)$, ciphertext $ct$ parsed as $(ct_1,ct_2)$, output $\mu$, where $\mu = ct_2 \oplus PRF_k(ct_1) \oplus otp$.
    \item \textbf{Fake Key Generation, $\fakegen(ct^0,m)$:} on input ciphertext $ct^0$ parsed as $(ct_1^0,ct_2^0)$, message $m$, it outputs the fake decryption key $fk= (k',otp')$, where $k'\leftarrow \{0,1\}^\secparam$ is uniformly sampled and $otp' = ct_2^0 \oplus PRF_{k'}(ct^0_1) \oplus m$.\\ {\em // Note: this choice of $otp'$ yields $\dec((k',otp'),ct^0) = m$.}
    
\end{itemize}

\paragraph{Correctness and Semantic Security:} Correctness can easily be checked. Semantic security follows from the security of pseudorandom functions using a standard argument.

\paragraph{Fake-Key Property.} Note that given $\{ct,(k,otp)\} \in \mathcal{C} \times \mathcal{K}$, one can perform the reversible operation: \begin{align*}
    \left\{(ct_1,ct_2), (k,otp)\right\} \longrightarrow \left\{(ct_1,ct_2 \oplus otp \oplus PRF_k(r)), (k,otp)\right\}.
\end{align*}

\noindent Thus, the fake-key property (\cref{fakekeyproperty}) can be rewritten as: \begin{align*}
    \{\left(ct^m \leftarrow \enc((k,otp),m),k \right)\} &\approx_c \{\left(ct^0 \leftarrow \enc((k,otp),0),fk \leftarrow \fakegen(ct^0, m) \right)\} \\
    \iff
    \{(r, PRF_k(r) \oplus m \oplus otp), (k,otp)\} &\approx_c \{(r, PRF_k(r) \oplus otp), (k',otp')\} \\
    \iff
    \{(r, m), (k,otp)\} &\approx_c \{(r, PRF_k(r) \oplus otp \oplus otp' \oplus PRF_{k'}(r)), (k',otp')\} \\
    \iff
    \{(r, m), (k,otp)\} &\approx_c \{(r, m), (k',otp')\} \\
    \iff
    \{(r, m), (k,otp)\} &\approx_c \{(r, m), (k', PRF_k(r) \oplus m \oplus otp)\} \\ 
    \iff
    \{(r, m), (k,otp)\} &\approx_c \{(r, m), (k, PRF_{k'}(r) \oplus m \oplus otp)\},  \numberthis \label{eight}
\end{align*}
where in the last step we swapped $k$ and $k'$, which is allowed since they are independently sampled. Therefore, observing in \cref{eight} that $k$ doesn't occur in the second part of the key, the fake-key property reduces to the following: \begin{align*}
    \{((r,m),otp)\} \approx_c \{((r,m),otp \oplus m \oplus PRF_{k'}(r))\},
\end{align*}
which follows\footnote{Note that this proof in fact demonstrates perfect (information-theoretic) fake-key property, even though we only need computational fake-key property in our construction.} from the fact that $otp$ is sampled independently from $r,m,$ and $k'$.

\end{proof}

\subsection{Construction}
We first describe the tools used in our construction of PK-UE scheme. 

\paragraph{Tools.} Let $\pke$ be a post-quantum private-key encryption scheme with fake-key property (defined in~\autoref{sec:privencfk}) and let $\ue$ be a one-time unclonable encryption scheme (defined in~\autoref{sec:otue}). \\

\noindent We present the construction of a PK-UE scheme below, which combines these tools such that it inherits semantic security from the first and unclonability from the second. \\

\paragraph{Setup, $\setup(1^{\secparam})$:} on input a security parameter $\secparam$, it outputs $k_{PKE}$, where 
$k_{PKE} \leftarrow \pke.\setup(1^\secparam)$.

\paragraph{Encryption, $\enc(k_{PKE}, m)$:} on input a key $k_{PKE}$, message $m$, it 
 first generates $k_{UE} \leftarrow \ue.\setup(1^\secparam)$ and outputs $ct = (ct_1, ct_2)$, where $ct_1 \leftarrow \pke.\enc(k_{PKE}, k_{UE})$ and $ct_2 \leftarrow \ue.\enc(k_{UE}, m)$.

\paragraph{Decryption, $\dec(k_{PKE}, ct)$:} on input the decryption key $k_{PKE}$, ciphertext $ct$, it computes $
     \mu = \ue.\dec(k_{UE}, ct_2)$, where $k_{UE} = \pke.\dec(k_{PKE}, ct_1)$. Output $\mu$.\\

\noindent Correctness follows from the correctness of the unclonable encryption scheme and the private-key encryption scheme. The semantic security follows from a standard hybrid argument and hence we omit the details; informally speaking, we first invoke the security of the underlying $\pke$ scheme to replace the message under $\pke$ to be 0 and then we invoke the indistinguishability security of $\ue$ to replace the message $m$. We perform this for all the $q$ messages, where $q=\mathrm {poly}(\secparam)$ is the number of messages chosen by the adversary in the semantic security experiment.

\subsubsection{Unclonable Security}

Suppose that for a parameter $t$, the proposed scheme is not $t$-unclonable secure; meaning there exists an adversary $A$ which breaks the corresponding cloning experiment (Hybrid 1) with probability $p = 2^{-n+t} + \frac{1}{\poly(\secparam)}$. We define another experiment Hybrid 2, which we claim the adversary breaks with probability $p - \negl(\secparam)$.

\paragraph{Hybrid 1:} The cloning experiment for $\pke$, the PK-UE scheme proposed above.

\paragraph{Hybrid 2:} \begin{itemize}
    \item In phase 1, the challenger samples $k_{PKE} \leftarrow \pke.\setup(1^{\secparam})$ and $k_{UE} \leftarrow \ue.\setup(1^{\secparam})$, then sends $(ct^0 \leftarrow \pke.\enc(k_{PKE}, 0), ct_2 \leftarrow \ue.\enc(k_{UE}, m))$ to the adversary $\alice$, who then applies a CPTP map $\phi: \density{A} \to \density{B} \otimes \density{C}$ to split it into two registers ($B,C$).
    \item In phase 2, the challenger reveals $fk \leftarrow \fakegen(ct^0, k_{UE})$ to both $\bob$ and $\charlie$, who then need to output $m_B = m_C = m$ in order to win the experiment. 
\end{itemize}

\begin{claim}
If $A$ wins in Hybrid 2 with probability $p'$, then $|p - p'| = \negl(\secparam)$.
\end{claim}

\begin{proof}
Assume to the contrary that $|p - p'| \ge \frac{1}{\poly(\secparam)}$. We will describe an adversary $\Tilde{\alice}$ which breaks the fake-key property of $\pke$.\\

\noindent Given $(ct^*, k^*_{PKE})$, $\Tilde{\alice}$ samples $k_{UE} \leftarrow \ue.\setup(1^{\secparam})$, computes $ct^m \leftarrow \ue.\enc(k_{UE}, m)$ and sends $(ct^*, ct^m)$ to $A$, who then applies a CPTP map $\phi: \density{A} \to \density{B} \otimes \density{C}$ to split it into two registers ($B,C$). In phase 2, $\Tilde{\alice}$ reveals $k^{*}_{PKE}$ to $\bob$ and $\charlie$. Observe that depending on whether the key $k_{PKE}^*$ is \textit{real} or \textit{fake}, we are either in Hybrid 1 or Hybrid 2. Hence, by assumption $\Tilde{\alice}$ can distinguish the two cases, breaking the fake-key property. 
\end{proof}

\noindent Now that we know $\alice$ breaks Hybrid 2 with probability at least $p - \negl(\secparam)$, we can construct an adversary $\vardbtilde{\alice}$ that breaks the unclonability experiment of $\ue$:

\begin{itemize}
    \item In Phase 1, the challenger samples $k_{UE} \leftarrow \ue.\setup(1^{\secparam})$ and  sends $ct^m \leftarrow \ue.\enc(k_{UE}, m)$ to $\vardbtilde{\alice}$. Then, $\vardbtilde{\alice}$ samples $k_{PKE} \leftarrow \pke.\setup(1^{\secparam})$ and computes $ct^0 \leftarrow \pke.\enc(k_{PKE}, 0)$. After that, $\vardbtilde{\alice}$ runs $A$ on input $(ct^0, ct^m)$ to obtain bipartite state $\rho_{BC} \in \density{B} \otimes \density{C}$, which she sends to $\vardbtilde{\bob}$ and $\vardbtilde{\charlie}$. In addition, $\vardbtilde{\alice}$ samples a randomness $r$ for the algorithm $\pke.\fakegen()$ and sends $r$ to both $\vardbtilde{\bob}$ and $\vardbtilde{\charlie}$.
    \item In phase 2, the challenger reveals $k_{UE}$ to both $\vardbtilde{\bob}$ and $\vardbtilde{\charlie}$. Then, $\vardbtilde{\bob}$ runs $\bob$ on his register\footnote{That is, the $B$ register of $\rho_{BC}$.}, revealing $fk$ as the key, to obtain and output $m_B$, where $fk \leftarrow \fakegen(ct^0,k_{UE})$ is sampled using randomness $r$. Similarly, $\vardbtilde{\charlie}$ obtains and outputs $m_C$ by running $\charlie$ on his register ($C$), revealing $fk$ as the key, where $fk$ is generated using randomness $r$ so that it matches what is generated by $\bob$.
    
\end{itemize}

Because the view of the adversary $\abc$ run as a subprotocol in this experiment matches exactly that in Hybrid 2, we conclude that $\vardbtilde{\alice}$ breaks the unclonability experiment of $\ue$ with probability $p'$, meaning $\ue$ is not $t-$unclonable secure.

Therefore, we just proved the following theorem.

\begin{theorem}
Assuming $\ue$ is a one-time unclonable encryption scheme with $t$-unclonable security, the encryption scheme constructed above is a private-key unclonable encryption scheme with $t$-unclonable security.
\end{theorem}

\begin{corollary}
The private-key unclonable encrpytion scheme proposed above is $n \log_2(1 + \frac{1}{\sqrt{2}})$-unclonable secure, where $n$ is the message length.
\end{corollary}
\newcommand{\prf}{\mathsf{PRF}}
\section{Public-Key Unclonable Encryption}
\label{sec:publickeyue}
\noindent We now focus on constructing unclonable encryption in the public-key setting using functional encryption. We adopt the Trojan technique of~\cite{ABSV15}, proposed in a completely different context, to prove the  unclonability property.  
\par We describe all the tools that we use in the scheme below. 

\paragraph{Tools.}
\begin{itemize}
    \item A one-time unclonable encryption scheme, denoted by $\ue=(\setup,\enc,\dec)$. 
    \item A post-quantum secure symmetric-key encryption scheme with pseudorandom ciphertexts, denoted by $\ske = (\setup,\enc,\dec)$. That is, this scheme has the property that the ciphertexts are computationally indistinguishable from the uniform distribution. Such a scheme can be constructed from one-way functions\footnote{The scheme is quite simple and presented in~\cite{GolVol1}: suppose $\prf:\{0,1\}^{\secparam} \rightarrow \{0,1\}^{\ell}$ is a pseudorandom function. To encrypt a message $x \in \{0,1\}^{\ell}$ using a symmetric key $k$, compute $(r,\prf(k,r) \oplus x)$, where $r \xleftarrow{\$} \{0,1\}^{\secparam}$. From the security of pseudorandom functions, it follows that the ciphertext is computationally indistinguishable from the uniform distribution.}. 
    \item A post-quantum secure single-key public-key functional encryption scheme, denoted by $\fe=(\setup,\keygen,\enc,\dec)$. Such a scheme can be instantiated using~\cite{SS10,GVW12}. See \autoref{sec:fe}.
\end{itemize}

\subsection{Construction}
We denote the public-key unclonable encryption scheme that we construct as $\pubkue = \allowbreak (\pubkue.\setup,\allowbreak \pubkue.\enc,\allowbreak \pubkue.\dec)$. We describe the algorithms below. 

\paragraph{Setup, $\setup(1^{\secparam})$:} On input a security parameter $\secparam$, 
compute $(\fe.\msk,\fe.\mpk) \leftarrow \fe.\setup(1^\secparam)$. Compute $\fe.sk \leftarrow \fe.\keygen(\fe.\msk,\allowbreak F[ct])$, where $ct \xleftarrow{\$} \{0,1\}^{\poly(\secparam)}$ and $F[ct]$ is the following function: 
$$F[ct](b,K,m) = \left\{ \begin{array}{cc} Dec(K,ct) & \text{if }b=0,\\ m, & \text{otherwise} \end{array} \right. $$
Set the secret key to be $k=\fe.sk$ and the public key to be $pk=\fe.\mpk$. 

\paragraph{Encryption, $\enc(pk,m)$:} On input key $pk$, message $m$, it first generates $k_{UE} \leftarrow \ue.\setup(1^\secparam)$, and outputs $ct = (ct_1, ct_2)$, where $ct_1 \leftarrow \fe.\enc(\fe.\mpk,(1,\bot,k_{UE})$ and $ct_2 \leftarrow \ue.Enc(k_{UE}, m)$.

\paragraph{Decryption, $\dec(k,ct)$:} On input $k$, ciphertext $ct=(ct_1,ct_2)$, first compute $\fe.\dec(\fe.sk,ct_1)$ to obtain $k^*_{UE}$. Then, compute $\ue.Dec(k^*_{UE},ct_2)$ to obtain $m^*$. Output $m^*$. \\

\noindent The correctness follows from the correctness of the underlying UE and FE schemes. As in the private-key setting, the semantic security follows by a standard argument and hence, we omit the details. 

\subsubsection{Unclonable Security}
We show that our construction achieves the same unclonable security as the underlying one-time scheme $\ue$. Formally, we prove the following theorem.

\begin{theorem}
If $\ue$ is $t$-unclonable secure, then $\pubkue$ is also $t$-unclonable secure.
\end{theorem}

\begin{proof}
Suppose that there exists an adversary $\abc$ which succeeds in the cloning experiment of $\pubkue$ with probability $p = 2^{-n+t} + \frac{1}{\poly(\secparam)}$. Through a sequence of hybrid experiments, we will construct an adversary which breaks the $t$-unclonability of $\ue$.

\paragraph{Hybrid 1:} This corresponds to the cloning experiment of $\pubkue$.

\paragraph{Hybrid 2:} Same as Hybrid 1, except $ct$ in $\pubkue.\setup()$, instead of being randomly sampled, is generated as $ct \leftarrow \ske.\enc(k_{SKE}, k_{UE})$, where $k_{SKE}\leftarrow \ske.\setup(1^{\secparam})$.

\begin{claim}
$\abc$ succeeds in Hybrid 2 with probability at least $p - \negl(\secparam)$.
\end{claim}
\begin{proof}
Hybrids 1 and 2 are computationally indistinguishable by the pseudorandom ciphertext property of $\ske$. Indeed, an adversary given a random text $r$ or a real ciphertext $ct$ can run the cloning experiment with $\abc$ to distinguish both the hybrids, hence distinguishing $r$ and $ct$.
\end{proof}

\paragraph{Hybrid 3:} Same as Hybrid 2, except $ct_1$ in $\pubkue.\enc()$, is generated as $ct_1 \leftarrow \fe.\enc(\fe.\mpk,\allowbreak  (0,k_{SKE}, \allowbreak \bot))$.

\begin{claim}
$\abc$ succeeds in Hybrid 3 with probability at least $p - \negl(\secparam)$.
\end{claim}
\begin{proof}
Hybrids 2 and 3 are indistinguishable by the (selective) security of $\fe$. Indeed, suppose that Hybrids 2 and 3 can be distinguished by $\abc$, and consider the following adversary $\alice'$ which breaks the (selective) security of $\fe$:

\begin{itemize}
    \item The challenger runs $(\fe.\mpk, \fe.\msk) \leftarrow \fe.\setup(1^\secparam)$. 
    \item $\alice'$ runs $k_{UE}\leftarrow \ue.\setup(1^\secparam)$ and $k_{SKE} \leftarrow \ske.\setup(1^\secparam)$, then sets $m_0 = (1,\bot,k_{UE})$ and $m_1 = (0,k_{SKE},\bot)$. Then, $\alice'$ sends $(m_0,m_1)$ to the challenger.
    \item The challenger chooses a random bit $b$ sends back $\fe.\mpk$ and $ct_1^b \leftarrow \fe.\enc(\fe.\mpk, m_b)$.
    \item $\alice'$ implements the function $\tilde{f} := F[\ske.\enc(k_{SKE}, k_{UE})]$ and makes a query to the challenger to receive $\fe.sk \leftarrow  \fe.\keygen(\fe.\msk, \tilde{f})$. This query is valid since $\tilde{f}(m_0) = \tilde{f}(m_1) = k_{UE}$. 
    
    \item Now $\alice'$ can perform a simulation, which matches Hybrid 2 with adversary $\abc$ when $b= 0$, and Hybrid 3 with adversary $\abc$ when $b = 1$. This will let $\alice'$ to distinguish the cases $b=0$ and $b=1$, breaking $\fe$ security. After sampling a random message $m \leftarrow \{0,1\}^n$, $\alice'$ has everything she needs to perform the simulation. Note that even though she doesn't know $\fe.\msk$, she has learned $\fe.sk$, which is the only time $\fe.\msk$ is used. 
    \end{itemize}

\end{proof}

Having established that $\abc$ succeeds in Hybrid 3 with probability $p - \negl(\secparam)$, we will now construct an adversary $\tldabc$ that succeeds in the cloning experiment of $\ue$ with probability $p - \negl(\secparam)$, contradicting the $t$-unclonable security:
\begin{itemize}
    \item The challenger samples $k_{UE} \leftarrow \ue.\setup(\secparam)$ and $m \leftarrow \{0,1\}^n$, then sends $ct_2 \leftarrow \ue.\enc(k_{UE},m)$ to $\vardbtilde{\alice}$.
    \item In Phase 1, $\vardbtilde{\alice}$ samples $(\fe.\msk, \fe.\mpk) \leftarrow \fe.\setup(1^\secparam)$ and $k_{SKE} \leftarrow \ske.\setup(1^\secparam)$. She then computes $ct_1 \leftarrow \fe.\enc(\fe.\mpk, (0, k_{SKE}, \bot))$. At the end of the phase $\vardbtilde{\alice}$ runs $\alice$ on input $ct^* = (ct_1, ct_2) $ to have $\vardbtilde{\bob}$ and $\vardbtilde{\charlie}$ receive bipartite state $\rho_{BC}\in \density{B} \otimes \density{C}$. $\vardbtilde{\alice}$ also samples a random string $r$ for $\ske.\enc$ and sends a copy of $r$ attached to the corresponding registers to both $\vardbtilde{\bob}$ and $\vardbtilde{\charlie}$.
    \item In Phase 2, the challenger reveals $k_{UE}$ to both $\vardbtilde{\bob}$ and $\vardbtilde{\charlie}$. $\vardbtilde{\bob}$ computes $ct \leftarrow \ske.\enc(k_{SKE}, k_{UE})$ (using randomness $r$), and $\fe.sk \leftarrow \fe.\keygen(\fe.\msk, F[ct])$. Then, he runs $\bob$ on the $B$ register of $\rho_{BC}$, revealing $\fe.sk)$ as the key, to obtain output $m_B$, which he outputs as is. Similarly, $\vardbtilde{\charlie}$ runs $\charlie$ to obtain and output $m_C$.
\end{itemize}
Described above, $\tldabc$ perfectly simulates the challenger of Hybrid 3 against $\abc$. Therefore, the success probability of $\tldabc$ is $p - \negl({\secparam})$. 

\end{proof}

\section{Additional Results on Unclonable Encryption} \label{sec:bounds}
\subsection{Generalized Conjugate Encryption}
\label{appendix:bounds1}
The conjugate encryption scheme of \cite{BL19} uses the BB84 monogamy-of-entanglement (MOE) game studied in
\cite{Tomamichel_2013}. The success probability of a cloning adversary exactly equals that of a MOE adversary restricted in state preparation. In this section we make the observation that their proof easily extends to a class of unclonable encryption schemes based on a class of MOE games, which we define below:

\begin{definition}[Real Orthogonal Basis] Let $(\ket{x}\bra{x})_{x \in X}$ be the standard basis for $\density{X}$, with $X = \{0,1,\dots,\dim \mathcal{H}_X - 1\}$. An orthonormal basis $\beta = (\ket{\psi_x}\bra{\psi_x})_{x \in X}$ for $\density{X}$ is called \textbf{real orthogonal} if there exist real coefficients $\left\{\alpha_{xx'}\right\}_{x,x' \in X}$ such that \begin{align*}
    \ket{\psi_x} = \sum_{x' \in X} \alpha_{xx'} \ket{x'}
\end{align*}
for all $x \in X$.
\end{definition} 

\noindent The following lemma, which is the main fact used to generalize conjugate encryption, states that an EPR pair defined in a real-orthogonal basis does not depend on the basis. It follows easily by properties of orthogonal matrices.
\begin{lemma}
  \label{lemma: real-orthogonal}
  If $\beta = (\ket{\psi_x}\bra{\psi_x})_{x \in X}$ is a real orthogonal basis, then
  \begin{align} \label{eq:basisindependence}
      \sum_{x \in X} \ket{\psi_x \psi_x} = \sum_{x \in X}{\ket{xx}}
  \end{align}
  and hence 
  \begin{align*}
      \sum_{x,x' \in X} \ket{x}\bra{x'} \otimes \ket{x}\bra{x'} = \sum_{x,x' \in X} \ket{\psi_x}\bra{\psi_{x'}} \otimes \ket{\psi_x}\bra{\psi_{x'}}
  \end{align*}
  by taking the outer product of each side by itself in \cref{eq:basisindependence}.
\end{lemma}

\begin{proof}
    By definition of a real orthogonal basis, the basis transition matrix $M = \left( \alpha_{xx'} \right)_{x,x' \in X}$ is an orthogonal matrix, and so is $M^T$. Thus, the columns of $M$ like its rows form an orthonormal basis, meaning \begin{align}
        \sum_{x \in X} \alpha_{xx'}\alpha_{xx''} = \delta_{x'x''}
        \end{align}
    for all $x',x'' \in X$. Hence,
    \begin{align*}
        \sum_{x \in X} \ket{\psi_x \psi_x} &= \sum_{x \in X} \left( \sum_{x' \in X} \alpha_{xx'} \ket{x'} \right) \left( \sum_{x'' \in X} \alpha_{xx''} \ket{x''} \right) \\
        &= \sum_{x,x',x'' \in X} \alpha_{xx'}\alpha_{xx''} \ket{x'x''} \\
        &= \sum_{x',x'' \in X} \delta_{x'x''} \ket{x'x''}\\
        &= \sum_{x \in X} \ket{xx}
    \end{align*}
\end{proof}
\begin{corollary} \label{cor:eprbasisindependence}
    If $X = \{0,1\}^n$ and $\beta = (\ket{\psi_x}\bra{\psi_x})_{x \in X}$ is a real orthogonal basis for $\density{X}$, then $$ \ket{\epr{n}}\bra{\epr{n}} = \sum_{x,x' \in X} \ket{\psi_x}\bra{\psi_{x'}} \otimes \ket{\psi_x}\bra{\psi_{x'}}. $$
\end{corollary}

\begin{definition}[Real-Orthogonal Monogamy Game\footnote{An example of a real-orthogonal game is studied in the context of coset states in \cite{CLLZ21hidden}. For a fixed subspace $A$ of $\mathbb{F}_2^n$ with dimension $n/2$, the coset states \[ \ket{A_{s,s'}} := \frac{1}{\sqrt{|A|}} \sum_{a \in A} (-1)^{\langle s', a \rangle} \ket{a + s} \] form a real-orthogonal basis for the $n$-qubit Hilbert space, where $s$ and $s'$ range over the cosets of $A$ and $A^\perp$, respectively. }] Let $X = \{0,1\}^n$.
A real-orthogonal monogamy game (ROMG) $\mathcal{G}$ of order $n$ is defined by the Hilbert space $\mathcal{H}_A$ of $n$-qubit states and a collection of real orthogonal bases $\left(\beta^\theta = \left(\ket{\psi_x^\theta}\bra{\psi_x^\theta}\right)_{x \in X}\right)_{\theta \in \Theta}$.
An adversary for $\mathcal{G}$ is defined by finite-dimensional Hilbert states $\mathcal{H}_B$ and $\mathcal{H}_C$, a tripartite state $\rho_{ABC} \in \density{A} \otimes \density{B} \otimes \density{C}$, along with two collections of POVMs: $\left(\left( B_x^\theta \right)_{x \in X}\right)_{\theta \in \Theta}$ and $\left(\left( C_x^\theta \right)_{x \in X}\right)_{\theta \in \Theta}$.
The value of $\mathcal{G}$, denoted by $p_G$, is the maximimum value the following expression can take for an optimal adversary: \begin{align*}
    p_{\text{win}} = \frac{1}{|\Theta|}\sum_{\theta \in \Theta} \tr \left( \Pi^\theta \rho_{ABC}\right),
\end{align*}
so that \begin{align*}
    p_G = \max_{\substack{\rho_{ABC} \in \density{A} \otimes \density{B} \otimes \density{C}\\ \left(\left( B_x^\theta \right)_{x \in X}\right)_{\theta \in \Theta}\\ \left(\left( C_x^\theta \right)_{x \in X}\right)_{\theta \in \Theta}}} p_{\text{win}},
\end{align*}
where \begin{align*}
    \Pi^\theta = \sum_{x \in X} \ket{\psi_x^\theta}\bra{\psi_x^\theta} \otimes B_x^\theta \otimes C_x^\theta.
\end{align*}
$p_{win}$ is the probability that $(\bob, \charlie)$ (the adversary) win in a monogamy game where: \begin{itemize}
    \item $\bob$ and $\charlie$, who are far away from each other, prepare a tripartite state $\rho_{ABC}$ and send the $A$ register to $\alice$. $\bob$ keeps the $B$ register and $\charlie$ keeps the $C$ register of this state.
    \item $\alice$ samples $\theta \in \Theta$ uniformly at random and measures her register in basis $\beta^\theta$ to obtain $x \in X$. She then sends $\theta$ to both $\bob$ and $\charlie$.
    \item $\bob$ and $\charlie$ guess the value $x$, and they win if they are both correct.
\end{itemize}
\end{definition}

\begin{theorem} \label{theorem:generalizedmoe}
Let $\mathcal{G}$ be a ROMG of order $n$ with value $p_G = 2^{-n+t} + \negl(\secparam
)$, then there exists an unclonable encryption scheme $\otue_G$ with (constant) message length $n$, which is $t$-unclonable secure.
\end{theorem}
\begin{proof}
We will construct a \textit{generalized conjugate encryption scheme} $\otue_G = (\setup, \enc, \dec)$ such that the success probability of a cloning adversary equals that of a ROMG adversary, which is bounded by $p_G$. The same construction and analysis is done by \cite{BL19} for the case of conjugate encoding \cite{wiesner83}, where $G$ is the BB84 game and $(\beta^\theta)$ are the Wiesner bases.\footnote{The BB84 game of order $n$ is defined as follows: $\beta^\theta = \left(\ket{\psi_x^\theta}\bra{\psi_x^\theta}\right)_{x \in X}$, where $$ \ket{\psi_x^\theta} = \bigotimes_{j=1}^n H^{\theta_j}\ket{x_j} $$ and $H$ denotes the single-qubit Hadamard gate.}\footnote{In \cite{BL19}, conjugate encryption is defined as having message length $n = \secparam$. We present $n$ to be a constant instead so that in the definition of $t$-unclonable security, the winning probability of a cloning adversary, which is negligible in $n$, is not negligible in $\secparam$.}
\paragraph{Setup:} On input security parameter $1^\secparam$, $\setup$ uniformly samples a key $(\theta,r) \leftarrow \Theta \times \{0,1\}^n $.
\paragraph{Encryption:} On input $m \in \mathcal{M}$ and $(\Theta,r) \in \mathcal{K}$, $\enc$ outputs the pure state $\rho = \ket{\psi_{(m \oplus r)}^\theta} \bra{\psi_{(m \oplus r)}^\theta}$.
\paragraph{Decryption:} On input ciphertext $\rho_{ct}$ and key $(\theta, r)$, $\dec$ measures $\rho_{ct}$ in the basis $\beta^\theta$ to obtain $x$, then outputs $x \oplus r$.

\paragraph{(One-Time) Indistinguishable Security:} It suffices to show that for any message, the view of an adversary with no knowledge of the key $(\theta, r)$ equals the completely mixed state, which can easily be done as \begin{align*}
    \frac{1}{2^\secparam|\Theta|}\sum_{\theta,r} \ket{\psi_{(m \oplus r)}^\theta} \bra{\psi_{(m \oplus r)}^\theta} = \mathbb{E}_\theta \frac{1}{2^\secparam} \sum_{x} \ket{\psi_x^\theta} \bra{\psi_x^\theta} = \mathbb{E}_\theta (\mathbf{id}/2^\secparam) = (\mathbf{id}/2^\secparam).
\end{align*}

\paragraph{$t$-unclonable security:} Let $\abc$ be a cloning adversary which uses the splitting CPTP map $\phi: \density{A} \to \density{B} \otimes \density{C}$ as well as POVMs $\left(B^{(\theta, r)}\right)_{\theta, r \in \{0,1\}^\secparam}$ and $\left(C^{(\theta, r)}\right)_{\theta, r \in \{0,1\}^\secparam}$. We will construct a ROMG adversary $\alice'$ for $\mathcal{G}$ that succeeds with the same probability. It uses the same Hilbert spaces $\mathcal{H}_B$, and $\mathcal{H}_C$, and it uses POVMs $(B')^\theta, (C')^\theta$, defined as $$(B')^{\theta}_x = \frac{1}{2^\secparam} \sum_{r \in \{0,1\}^\secparam} B_{x \oplus r}^{(\theta, r)}, \quad (C')_x^{\theta} = \frac{1}{2^\secparam} \sum_{r \in \{0,1\}^\secparam} C_{x \oplus r}^{(\theta, r)}.$$
Finally, the tripartite state $\rho_{ABC}$ is defined below using \autoref{cor:eprbasisindependence}: \begin{align*}
    \rho_{ABC} &= \left(\id_A \otimes \phi\right) \ket{\epr{\secparam}}\bra{\epr{\secparam}} \\ &= \left(\id_A \otimes \phi\right) \frac{1}{2^\secparam}\sum_{s,u\in \{0,1\}^\secparam} \ket{s}\bra{u} \otimes \ket{s}\bra{u} \\
    &= \left(\id_A \otimes \phi\right) \frac{1}{2^\secparam}\sum_{s,u\in \{0,1\}^\secparam} \ket{\psi_s^\theta}\bra{\psi_u^\theta} \otimes \ket{\psi_s^\theta}\bra{\psi_u^\theta} \\
    &= \frac{1}{2^\secparam}\sum_{s,u\in \{0,1\}^\secparam} \ket{\psi_s^\theta}\bra{\psi_u^\theta} \otimes \phi\left(\ket{\psi_s^\theta}\bra{\psi_u^\theta}\right).
\end{align*}
The success probability of $\alice'$ is then given by \begin{align}
    p_G &\ge \frac{1}{|\Theta|}\sum_{\theta \in \Theta} \sum_{x \in \{0,1\}^\secparam} \tr \left[ \left( \ket{\psi_x^\theta}\bra{\psi_x^\theta} \otimes (B_x')^\theta \otimes (C_x')^\theta \right) \rho_{ABC} \right] \nonumber \\
    &= \frac{1}{2^{2\secparam}|\Theta|}\sum_{\theta \in \Theta} \sum_{x,r,s,u \in \{0,1\}^\secparam} \tr \left[ \left( \ket{\psi_x^\theta}\bra{\psi_x^\theta} \otimes B_{x \oplus r}^{(\theta, r)} \otimes C_{x \oplus r}^{(\theta, r)} \right) \left( \ket{\psi_s^\theta}\bra{\psi_u^\theta} \otimes \phi\left(\ket{\psi_s^\theta}\bra{\psi_u^\theta}\right) \right) \right] \nonumber \\
    &= \frac{1}{2^{2\secparam}|\Theta|}\sum_{\theta \in \Theta} \sum_{x,r \in \{0,1\}^\secparam} \tr \left[ \left( B_{x \oplus r}^{(\theta, r)} \otimes C_{x \oplus r}^{(\theta, r)} \right) \phi\left(\ket{\psi_x^\theta}\bra{\psi_x^\theta}\right) \right] \label{eq:prev}\\
    &= \E_{m, \theta, r} \tr \left[ \left( B_{m}^{(\theta, r)} \otimes C_{m}^{(\theta, r)} \right) \phi\left(\ket{\psi_{m\oplus r}^\theta}\bra{\psi_{m\oplus r}^\theta}\right) \right] \label{eq:winningprob}
    \end{align}
After putting $x = m\oplus r$ in (\ref{eq:prev}), we see that (\ref{eq:winningprob}) above equals the winning probability of $\abc$ in the cloning experiment, and it is bounded by $p_G = 2^{-n+t} + \negl(\secparam)$ which suffices for the proof.
\end{proof} 

\noindent We are not aware of a MOE game with value provably less than $(1/2 + 1/2\sqrt{2})^n$, nor are we aware of a proof that it does not exist. Nevertheless, any advancement on this front will give insight to optimal unclonable-security by \autoref{theorem:generalizedmoe}.

\subsection{A Lower Bound for Conjugate Encryption.}
\label{appendix:bounds2}
 A natural question to explore is whether 0-unclonable security\footnote{\cite{BL19} show that 0-unclonable security implies unclonable-indistinguishable security, making this question more interesting.} is possible, even for single-bit messages, since 0-unclonable security means that a cloning adversary $\abc$ does not benefit from cloning the ciphertext at all, and hence cannot do better than the trivial strategy of giving the ciphertext to $\bob$ and having $\charlie$ randomly guess the message. In this section we show that the conjugate encryption of \cite{BL19} is not $0-$unclonable secure. To show this, we note that the valid ciphertexts in conjugate encryption for one-bit messages all lie on the $xz$-plane of the Bloch Sphere, i.e. they do not have an imaginary phase in the computational basis. Besides, encrypting multi-bit messages is done simply by encrypting each bit separately. The following lemma, which refers to the optimal equatorial cloner studied in \cite{Bru__2000}, will take advantage of this fact:
 \begin{lemma}
   \label{lemma:cloning} Let $\mathcal{D} =  \density{2}$ denote the space of one-qubit states. Then, there exists a cloning map $\Phi: \mathcal{D} \to \mathcal{D} \otimes \mathcal{D}$ such that $F(\rho, \tr_C(\Phi(\rho))) \ge 1/2 + 1/2\sqrt{2}$ and $F(\rho, \tr_B(\Phi(\rho))) \ge 1/2 + 1/2\sqrt{2}$ for any $\rho$ which is a valid ciphertext of a single-bit message in generalized conjugate encryption, where $\tr_X$ is the partial trace operation of tracing out the $X$ register.
 \end{lemma}
 The following result, then is imminent:
 \begin{theorem} \label{thm:lowerbound}
 Let $G$ be a real-orthogonal monogamy game of order $n$ which is an $n$-fold parallel repetition of a real-orthogonal monogamy game of order 1, i.e. the basis states are of the form \[ \ketbra{\psi^\theta_x}{\psi^\theta_x} = \bigotimes_{i=1}^n \ketbra{\psi^{\theta_i}_{x_i}}{\psi^{\theta_i}_{x_i}}, \]
     where $\Theta = (\widetilde{\Theta})^n$, $\theta = (\theta_1, \dots, \theta_n)$, and $\left( \ketbra{\psi^{\theta_i}_{b}}{\psi^{\theta_i}_b}  \right)_{b \in \{0,1\}}$ is a real-orthogonal basis of the one-qubit Hilbert space for any $\theta_i \in \widetilde{\Theta}$. \\ \\
\noindent Then, the generalized conjugate encryption $\otue_G$ as defined in the proof of \autoref{theorem:generalizedmoe} is \textbf{not} $(cn)$-unclonable secure for any constant $c < 1/2$.
 \end{theorem}
  
 \begin{proof}
 It suffices to construct a cloning adversary adversary $\abc$ which succeeds with probability $2^{-n/2}$. At a high level, we do the following: since every qubit is encrypted individually, $\bob$ and $\charlie$ will independently guess each qubit of the message. 
 \par By \autoref{lemma:cloning}, there exists a cloner $\Phi: \density{2} \to \density{2} \otimes \density{2}$ which clones every qubit of a valid ciphertext $\rho_{\ct}$ with fidelity $f = 1/2 + 1/2\sqrt{2}$. Given a ciphertext $\rho_{\ct}$ in phase 1, $\alice$ will use the map $\Phi^{\otimes n}$ to split it into two registers of $n$-qubits, so that if $\rho_B = \tr_C(\Phi^{\otimes n}\left(\rho_\ct \right)) = \bigotimes_{i=1}^n\rho_{B,i}$ is the local view of $\bob$, then $F(\rho_{B,i},\rho_{\ct,i}) \ge f$ (similarly for $\charlie$).\\ \\
 \noindent In phase 2, after the key $k$ is revealed, $\bob$ and $\charlie$ each apply $\dec(k, .)$ to their register, which can be applied qubit-wise. Since fidelity cannot decrease with quantum operations \autoref{lem:fidelity}, the local view of $\rho_{B,i}'$ of $\bob$ after decrypting has fidelity at least $f$ to $\ket{m_i}\bra{m_i} = \dec(k, \rho_{ct})$, meaning $\braket{m_i}{\rho_{B,i}' | m_i} \ge f$ (similarly for $\charlie$). \\ \\
 \noindent Next, $\bob$ and $\charlie$ measure their register in the standard basis. By definition of fidelity, then $\Pr[m_{B,i} = m_i] \ge f$ and $\Pr[m_{C,i} = m_i] \ge f$. By union bound, this implies $\Pr[m_{B,i} = m_{C,i} = m_i] \ge 2f - 1 = 2^{-1/2}$. Since every bit of the message $m$ is independent, it follows that $\Pr[m_B = m_C = m] \ge (2f-1)^n = 2^{-n/2}$ as desired.
 \end{proof}
 
 The bound in \autoref{thm:lowerbound}, which states that a cloning adversary cannot succeed with probability greater than $0.71^n$, applies to conjugate encryption of \cite{BL19}. Yet, there is an even simpler cloning attack which targets this scheme specifically and succeeds with probability $0.75^n$.
 
 \begin{theorem} \label{thm:lowerbound2}
 Conjugate encryption is not $(cn)$-unclonable secure for any $c < 1 - \log_2(3/4)$.
 \end{theorem}
 \begin{proof}
 Consider the following cloning adversary $\abc$: given a ciphertext $\rho = \enc((\theta,r),m) = H^\theta\ket{m}\bra{m}H^\theta$, $\alice$ samples $\theta' \uniform \{0,1\}^n$ and measures $\rho$ in the Wiesner basis $H^{\theta'}$ to obtain a classical string $x \in \{0,1\}^n$, which she sends to both $\bob$ and $\charlie$. After the key $(\theta, r)$ is revealed, $\bob$ and $\charlie$ each output $m' := x \oplus r$. For each $i$, we have $m_i' = m_i$ with probability 1 conditioned on $\theta_i' = \theta_i$, and $m_i' = m_i$ with probability $1/2$ conditioned on $\theta_i' \ne \theta_i$. Therefore, the success probability of this adversary is given by \[ \prob \left[ m' = m \right] =  \prod_{i=1}^n \prob \left[ m_i' = m_i\right] = \prod_{i=1}^n \left(\frac{1}{2} \prob \left[ m_i' = m_i \given \theta_i' = \theta_i\right] + \frac{1}{2} \prob \left[ m_i' = m_i \given \theta_i' \ne \theta_i\right]\right) =  \left(\frac{3}{4}\right)^n \]
 
 \end{proof}

\section{Construction of Copy-Protection from Unclonable Encryption} \label{sec:copyprotection}
In this section, we present an application of unclonable encryption by showing that the existence of an unclonable-indistinguishable secure scheme (see \autoref{def:unclonableindistinguishable}) implies a copy-protection scheme over a special class of point functions. unclonable-indistinguishable security seems to be a stronger notion than unclonable security, and it remains open question whether it is possible. \\
\noindent The main drawback of our construction is that the copy-protected program $\rho_f$ for the point function $f_{a,b}$ is reusable only if it is used to evaluate the function on the "correct" input $a$. When $f$ is evaluated on inpnuts $x \ne a$, our scheme does not guarantee that $\rho_f$ will not be destroyed.

\paragraph{Construction:} Let $(\gen, \sign, \ver)$ be a post-quantum signature scheme, and let $\ue$ be an unclonable-indistinguishable secure unclonable-encryption scheme encrypting $n$-bit messages, where $n$ is the size of a signature created by $\sign()$. We construct a copy-protection scheme ($\cp, \eval$) for the family $\mathcal{F} = \{f_{k,(vk||\sigma)}: k \leftarrow \ue.\setup(1^\secparam), (vk, sk) \leftarrow \gen(1^\secparam), \sigma \leftarrow \sign(sk, 0)\}$ of point functions. Let $X,Y$ denote the domain and codomain of $f \in \fclass$.

\begin{itemize}
    \item{\bf Copy-Protected State Generation:} On input the security parameter $1^\secparam$ and description $(a,b)$ of a function $f_{a,b} \in \fclass$, $\cp$ does the following:
    \begin{itemize}
        \item Parse $a$ as $k$ and $b$ as $vk||\sigma$. 
        \item Compute $\rho \leftarrow \ue.\enc(k, \sigma)$.
        \item Output $\tilde{\rho}= \rho \otimes \ket{vk}\bra{vk}$. 
    \end{itemize}
    \item{\bf Evaluation:} On input the security parameter $1^\secparam$, a value $x\in X$ and a copy-protected state $\tilde{\rho}$, $\eval$ does the following:
    \begin{enumerate}
        \item Measure the second register of $\tilde{\rho}$ to obtain the state $\rho \otimes \ket{vk}\bra{vk}$
        \item Compute $\sigma' \leftarrow \ue.\dec(x, \rho)$ and $\delta \leftarrow \ver(vk, 0, \sigma')$. 
        \item If $\delta = 0$, set $y = 0$; if $\delta = 1$, set $y = vk||\sigma'$. Output $ \ket{vk}\bra{vk} \otimes \ket{y}\bra{y}$.
    \end{enumerate}
    Using \autoref{cor:uncomputing}, we can reimplement the second and third steps above so that $\eval$ outputs a state $(\rho' \otimes \ket{vk}\bra{vk})\otimes 
    \ket{y}\bra{y}$, where $\rho'$ is close to $\rho$ on correct inputs. We assume that $\eval$ does this for reusability purposes.
\end{itemize}
\paragraph{Computational Correctness:} Our construction satisfies computational correctness  (\autoref{def:computationalcorrectness}). In order to find an input that evaluates incorrectly, an adversary must be able to forge a signature using only the verification key $vk$. We formalize this argument below:

\begin{claim}
Assuming unclonable-indistinguishability property of UE and the unforgeability property of the unique signature scheme, $(\cp, \eval)$ satisfies computational $(\negl(\secparam), \negl(\secparam))$-correctness.
\end{claim}
\begin{proof}
The first bullet point of the computational correctness property follows from the statistical correctness of $\ue$ and \autoref{cor:uncomputing}. This is because when the function is evaluated with the correct key ($a = k$), the decryption succeeds with probability $1 - \negl(\secparam)$, which implies that it is (almost) reversible. \\ \\
\noindent We prove the second bullet via proof by contradiction. Consider the following hybrid experiments: \\

\noindent $\hybrid_1$: This corresponds to the real correctness experiment, where the adversary receives as input a copy-protection of the point function $f_{a,b}$ and needs to find a value $x' \neq a$ such that the evaluation of the copy-protected state on the input $x'$ yields a non-zero value. 
Let the success probability of $\alice$ in this experiment be $\varepsilon_1$. \\

\noindent $\hybrid_2$: This hybrid is identical to $\hybrid_1$, except that we change the way we are computing the $\ue$ ciphertext. Instead of computing $\rho \leftarrow \ue.\enc(k,\sigma)$, we compute $\rho \leftarrow \ue.\enc(k,0)$. Let the success probability of $\alice$ in this experiment be $\varepsilon_2$. \\

\par We first argue that $|\varepsilon_1 - \varepsilon_2| \le \negl(\secparam)$. To prove this, we will construct an adversary $\alice'$ which tries to break the one-time indistinguishable security of $\ue$: \begin{itemize}
    \item $\alice'$ samples $(vk, sk) \leftarrow \gen(1^\secparam)$ and computes $\sigma \leftarrow \sign(sk, 0)$. She then sends two messages $m_0 = \sigma$ and $m_1 = 0$ to the challenger.
    \item The challenger samples $k \leftarrow \ue.\setup(1^\secparam)$ and a uniformly random bit $b$. He sends $\rho_\ct \leftarrow \ue.\enc(k, m_b)$ to $\alice'$.
    \item $\alice'$ sets $\rho_f^{(0)} = \rho_\ct \otimes \ket{vk}\bra{vk}$ and simulates the correctness experiment corresponding to $f_{k, (vk || \sigma)}$ by running $\alice$ and playing the role of the challenger in that experiment. She outputs 1 if $\alice$ succeeds; otherwise, she outputs 0.
\end{itemize} 

\noindent If $b = 0$, then $m_b = \sigma$ and the view of $\alice$ is $\hybrid_1$. Hence, $\alice'$ outputs 1 with probability $\varepsilon_1$. \\
\noindent On the other hand, if $b = 1$, then $m_b = 0$ and the view of $\alice$ is $\hybrid_2$, so that $\alice'$ outputs 1 with probability $\varepsilon_2$. \\
\noindent Therefore, by one-time indistinguishable security of $\ue$, it follows that $|\varepsilon_1 - \varepsilon_2| \le \negl(\secparam)$. \\
\par Secondly, we argue that $\varepsilon_2 \le \negl(\secparam)$ by constructing an adversary $\forger$ which tries to break the unforgeability property of the signature scheme: \begin{itemize}
    \item The challenger samples $(vk, sk) \leftarrow 
    \gen(1^\secparam)$ and sends $vk$ to $\forger$. 
    \item $\forger$ samples $k \leftarrow \ue.\setup(1^\secparam)$ and computes $\rho_\ct\leftarrow \ue.\enc(k, 0)$. He sets $\rho_f^{(0)} = \rho_\ct \otimes \ket{vk}\bra{vk}$ and simulates $\hybrid_2$ by running $\alice$ and playing the role of the challenger in that experiment. If there exists a query $x_i$ such that the answer $y_i = vk' || \sigma'$ to that query satisfies $y_i \ne 0$, then $\forger$ outputs $\sigma'$; otherwise, $\forger$ outputs $0$.
\end{itemize}
With probability $\varepsilon_2$, $\alice$ will succeed in the experiment $\hybrid_2$, and $\alice'$ will output $\sigma'$ such that $y_i = vk' || \sigma' $, where $\rho_f^{(1)} \otimes \ket{y_i}\bra{y_i} \leftarrow \eval(1^\secparam, \rho_f^{(i)}, x_i)$ for a query $x_i$.
\par Note that in our construction, even though the states $\left(\rho_f^{(j)}\right)_{j=0}^{\poly(\secparam)}$ could be different, they preserve the initial verification key $vk$. Hence, $\eval$ always runs signature verification using $vk$. Therefore, $\forger$ outputs a valid signature $\sigma'$ on $0$ with probability $\varepsilon_2$, so it is negligible by the unforgability of the signature scheme.

\end{proof}

\paragraph{Copy-Protection Security:}

\begin{claim}
The construction above is a secure copy-protection scheme assuming the one-time existential unforgeability property of the signature scheme $(\gen, \sign, \ver)$ and the unclonable-indistinguishable security of the unclonable encryption scheme $\ue$. 
\end{claim}
\begin{proof}
Suppose there exists an adversary $\abc$ that breaks the copy-protection security of our construction (see \autoref{def:copyprotectionpointfunction}). Let $\hybrid_1$ be the corresponding pirating experiment, where the challenger always sends $x = k$ (in the original pirating experiment, he sends this input only half the time). It follows that with non-negligible probability $p$, both $\bob$ and $\charlie$ output $(vk || \sigma)$ in $\hybrid_1$. In other words, given a copy-protected program $\rho_f$ for a point function $f_{k, (vk || \sigma)}$, $\alice$ can prepare a bipartite state on registers $B$ and $C$ such that on input $k$, both $\bob$ and $\charlie$ output $\sigma$ with probability $p$. (We ignore $vk$ in the output for simplified notation in this proof.) \\ \\

\noindent We define a new experiment $\hybrid_2$, which is identical to $\hybrid_1$ except when the challenger is computing the copy-protected state $\rho_f = \ue.\enc(k, \sigma) \otimes \ket{vk}\bra{vk}$, he insteads computes $\rho_f' = \ue.\enc(k, 0) \otimes \ket{vk}\bra{vk}$ and sends it to $\alice$. \\ \\ 
We first argue that in $\hybrid_2$, the probability that either $\bob$ or $\charlie$ outputs $\sigma$ is negligible in $\secparam$. This follows from the fact that if w.l.o.g. $\bob$ outputs $\sigma$ with non-negligible probability, then there exists an adversary $\forger$ which breaks the unforgeability of the signature scheme: \begin{itemize}
    \item $\forger$ Given the security parameter $1^\secparam$ and $vk$ such that $(vk, sk) \leftarrow \gen(1^\secparam)$, $\forger$ samples a key $k \leftarrow \ue.\setup(1^\secparam)$ and computes $\rho \leftarrow \ue.\enc(k, 0)$.
    \item $\forger$ then runs $\abc$ by sending $\rho_f' = \rho \otimes \ket{vk}\bra{vk}$ to $\alice$ and simulating the experiment $\hybrid_2$. It outputs the output of $\bob$, which is a valid signature on $0$ with non-negligible probability.
\end{itemize}

\par Now we construct a cloning-distinguishing adversary $\abcprime$ which breaks the unclonable-indistinguishable security of $\ue$: \begin{itemize}
    \item In phase 1, $\alice'$ samples $(vk, sk) \leftarrow \gen(1^\secparam)$ and computes $\sigma \leftarrow \sign(sk, 0)$. She then sends messages $m_0 = \sigma$ and $m_1 = 0$ to the challenger.
    \item In phase 2, the challenger computes $\rho_{\ct} = \ue.\enc(k, m_b)$ for $k \leftarrow \ue.\setup(1^\secparam)$ and a uniformly random bit $b$. He sends $\rho_{\ct}$ to $\alice'$.
    \item $\alice'$ runs $\alice$ by sending $\rho_{\ct} \otimes \ket{vk}\bra{vk} $ as the copy-protected program and to create a bipartite state over registers $B,C$. She sends the $B$ ( resp., $C$) register to $\bob'$ (resp., $\charlie'$).
    \item In phase 3, the key $k$ is revealed to $\bob'$ and $\charlie'$. $\bob'$ then runs $\bob$ as if $x_B = k$ in the pirating experiment, similarly for $\charlie'$. Note that if $b=0$, the view of $\bob$ and $\charlie$ is exactly $\hybrid_1$ and if $b=1$ it is $\hybrid_2$. Let the output of $\bob$ and $\charlie$ be $y_B$ and $y_C$, respectively. In the end, $\bob'$ (resp., $\charlie'$) outputs the bit $b_B = 0$ if and only if $y_B = \sigma$ (resp., $y_C = \sigma$).
\end{itemize}
The probability that $\bob'$ and $\charlie'$ simultaneously predict the bit $b$ correctly is given by $$ \frac{1}{2} \left(\Pr[y_B = y_C = \sigma \given b = 0] + \Pr[y_B \ne \sigma \land y_C \ne \sigma \given b = 1]\right) \ge \frac{1}{2}(p + 1 - \negl(\secparam)) \ge \frac{1}{2} + \frac{p}{2} - \negl(\secparam), $$
thus breaking the unclonable-indistinguishable security.

\end{proof}

\bibliographystyle{alpha}
\bibliography{crypto}

\appendix

\nocite{PhysRevLett.79.2153}

\end{document}